%% file: main.tex
\PassOptionsToPackage{nameinlink}{cleveref}
\documentclass[a4paper,UKenglish,cleveref, autoref,numberwithinsect, thm-restate]{lipics-v2021}

\hideLIPIcs
\nolinenumbers

\usepackage{mathtools}
\usepackage{dsfont}
\usepackage{booktabs}
\usepackage{tablefootnote}
\usepackage{float}
\usepackage{pgfplots}
\pgfplotsset{compat=1.18}
\usepackage[ruled,vlined,linesnumbered,noend]{algorithm2e}
\SetKwInOut{Parameter}{Parameters}
\usepackage{dsfont}

\usepackage{url}
\usepackage[numbers]{natbib}
\newcommand{\ncite}[1]{\citeauthor*{#1}~\cite{#1}}

\AtBeginDocument{%
  \definecolor{ForestGreen}{rgb}{0.1333,0.5451,0.1333}
  \definecolor{DarkRed}{rgb}{0.65,0,0}
  \definecolor{myred}{HTML}{CC3311}
  \definecolor{myblue}{HTML}{0077BB}
}

\DeclareFontShape{T1}{lmr}{m}{scit}{<->ssub * lmr/m/scsl}{}

\hypersetup{colorlinks=true, linkcolor=DarkRed, filecolor=magenta, urlcolor=cyan, citecolor=ForestGreen}

\crefname{ineq}{}{}
\crefname{equation}{}{}
\Crefname{ineq}{Inequality}{Inequalities}
\crefformat{ineq}{#2inequality~(#1)#3}
\Crefformat{ineq}{#2Inequality~(#1)#3}

\newcounter{property}
\crefname{property}{Property}{Properties}
\Crefname{property}{Property}{Properties}

\DeclareMathOperator*{\vol}{vol}

\DeclareMathOperator*{\U}{\mathcal U}

\DeclarePairedDelimiter{\set}{\{}{\}}

\newcommand{\Alg}{\textsc{Clever-Greedy}\xspace}

\newcommand{\e}{\mathrm{e}}
\newcommand{\OPTAlg}{\textsc{Opt}\xspace}
\newcommand{\OPT}[1]{\textsc{Opt}(#1)}
\newcommand{\OPTk}[1]{\theta_{#1}}
\newcommand{\define}{\coloneqq}
\newcommand{\R}{\mathds{R}}

\newcommand{\Z}{\mathds{Z}}

\title{Incremental Submodular Maximization:\\Better Than Greedy}
\titlerunning{Incremental Submodular Maximization: Better Than Greedy}

\author{Marcin Bienkowski}{University of Wroc{\l}aw, Poland}{marcin.bienkowski@cs.uni.wroc.pl}{0000-0002-2453-7772}{Supported by Polish National Science Centre grant 2022/45/B/ST6/00559}
\author{Joakim Blikstad}{CWI Amsterdam}{blikstad@kth.se}{0009-0004-0874-2356}{}
\author{Jaros{\l}aw Byrka}{University of Wroc{\l}aw, Poland}{jby@cs.uni.wroc.pl}{}{Supported by Polish National Science Centre grant 2020/39/B/ST6/01641}
\author{Mart\'in Costa}{University of Warwick, UK}{martin.costa@warwick.ac.uk}{}{}
\author{Yann Disser}{TU Darmstadt, Germany}{disser@mathematik.tu-darmstadt.de}{0000-0002-2085-0454}{Supported by Deutsche Forschungsgemeinschaft (DFG, German Research Foundation) through subproject A09 of CRC/TRR154.}
\author{Annette Lutz}{TU Darmstadt, Germany}{lutz@mathematik.tu-darmstadt.de}{0009-0008-7699-7018}{Supported by Deutsche Forschungsgemeinschaft (DFG, German Research Foundation) through subproject A09 of CRC/TRR154.}

\authorrunning{M. Bienkowski, J. Blikstad, J. Byrka, M. Costa, Y. Disser, and A. Lutz}

\Copyright{Marcin Bienkowski, Joakim Blikstad, Jaros{\l}aw Byrka, Mart\'in Costa, Yann Disser, and Annette Lutz}



\funding{This research has been initiated at the AlgUW Workshop (Bedlewo 09.2025), supported by the Excellence Initiative -- Research University (IDUB) funds of the University of Warsaw.}

\usepackage{hyperref}
\usepackage{doi}

\makeatletter
\newcommand*{\addSectionToHyperref}[1]{%
	\expandafter\renewcommand\csname theH#1\endcsname{\thesection.\arabic{#1}}%
}
\addSectionToHyperref{lemma}
\addSectionToHyperref{theorem}
\addSectionToHyperref{proposition}
\addSectionToHyperref{property}
\addSectionToHyperref{claim}
\makeatother

\begin{document}

\maketitle

\begin{abstract}
We consider submodular maximization under increasing cardinality constraint and ask for a good incremental solution, i.e., an ordering of the ground set such that each prefix of the ordering yields a good solution for its respective cardinality.
A classical result in this setting is that the greedy algorithm achieves a competitive ratio, i.e., an approximation guarantee across all cardinalities, of~$\e/(\e-1) \approx 1.582$.
No better general guarantee was previously known.
We present an adaptive scaling algorithm achieving a competitive ratio of~$1.373$.
We complement our result by a deterministic lower bound of~$1.25$ on the best possible competitive ratio for incremental submodular maximization.
\end{abstract}

\newpage


\section{Introduction}

We consider cardinality constrained maximization problems of the form
\begin{align} 
S_k^\star \define \arg\max \{ f(S) \mid S \subseteq \U, |S| \leq k \}, \label{eq:main}
\end{align}
with a monotone%
\footnote{$f\colon 2^{\U} \to \R_{\geq 0}$ is monotone if $A\subseteq B \subseteq \U \Longrightarrow f(A) \leq f(B)$.}, submodular%
\footnote{$f\colon 2^{\U} \to \R_{\geq 0}$ is submodular if $A,B \in \U \Longrightarrow f(A) + f(B) \geq f(A\cup B) + f(A \cap B)$.}
and non-negative objective $f\colon 2^{\U} \to \R_{\geq{0}}$ over a finite ground set~$\U$ of size $n \define |\U|$, and a given cardinality~$k \in [n] \define \{1,\dots,n\}$.
It is a classical result that the greedy algorithm, that extends the solution~$S$ by the element~$e \in \U$ maximizing $f(S \cup \{e\})$ in each step, produces a solution that is within a~factor of $\frac{\e}{\e - 1} \approx 1.582$ of the optimum~\cite{NemhauserWolseyFisher1978}, and that this is best-possible for polynomial time algorithms~\cite{FeigeMirrokniVondrak2011,Vondrak2013}, even if~$f$ is a coverage function~\cite{Feige1998} (assuming $\mathrm{P} \neq \mathrm{NP}$).

A major advantage of the greedy algorithm, besides being natural, simple and efficient, lies in the fact that it produces an \emph{incremental} solution, i.e., that its solutions to different cardinalities share the same prefix in terms of the order in which elements are included.
This means that greedy solutions can be gradually implemented as the cardinality bound increases, which is relevant, e.g., in large infrastructure projects that are implemented over time, but must already be useful long before their completion.
 
This raises the natural question whether better incremental solutions than greedy are possible if we do not require efficient computation of the solution.
Sacrificing efficiency for quality of the solution is particularly well-motivated for long-term projects where computation time is short compared to the implementation timeframe of the solution.
Also, outside the worst case, incremental solutions better than the greedy guarantee may sometimes be possible to compute in polynomial time.

We adopt the framework of incremental maximization introduced by \ncite{BernsteinDisserGrossHimburg-20} to address this question.
Formally, an \emph{incremental solution} is given by a sequence of nested solutions $\emptyset = S_0 \subseteq \dots \subseteq S_n = \U$, where each~$S_i$ with $|S_i| = i$ is a~solution to \cref{eq:main} for cardinality~$k = i$. 
We can also identify an incremental solution with an ordering $\{e_1,\dots,e_n\} = \U$ of the elements in the ground set $\U$, where $e_i = S_i\setminus S_{i-1}$. 
We say that the incremental solution is \emph{$\rho$-competitive} if
\[ 
    \max_{k \in \{1,\dots,n\}} \frac{f(S_k^\star)}{f(S_k)} \leq \rho,
\]
and define the \emph{competitive ratio}\footnote{We use the notion of \emph{competitive} ratio in contrast to \emph{approximation} ratio to indicate that we do not require polynomial time. Incremental maximization can also be phrased as an online problem, where the only information revealed online in step~$i$ is whether or not~$k \geq i$ (much like in ski rental).} of an algorithm as the infimum over all $\rho\geq 1$ such that the algorithm guarantees a $\rho$-competitive solution.

In these terms, the greedy algorithm is well known to have competitive ratio~$\frac{\e}{\e - 1}$. We ask for the best possible competitive ratio for incremental submodular maximization, i.e., for the \emph{price of incrementality} that quantifies the loss in solution quality that we have to accept if we require a consistent solution for all cardinalities simultaneously.

\subparagraph*{Our Results.}

Our main result is the first algorithm (cf.~\Cref{alg:main_simplified}) to beat the greedy competitive ratio of $\e/(\e-1) \approx 1.582$ for the incremental submodular maximization problem.

\begin{restatable}{theorem}{mainupper}
    \label{thm:main-upper}
    \Alg has a competitive ratio of $\rho \approx 1.373$ for incremental submodular maximization.
\end{restatable}

In particular, our analysis is tight, up to the numerical precision of our analysis.

Note that our algorithm is not efficient but rather provides a proof that a good incremental solution always \emph{exist}, i.e., a good consistent solution for all cardinalities.
In fact, no better approximation algorithm than greedy for submodular maximization under even a fixed cardinality constraint is possible with a sub-exponential number of value-oracle accesses to~$f$~\cite{FeigeMirrokniVondrak2011,Vondrak2013}.
In accordance to this, our algorithm requires access to optimum solutions for each cardinality.
While this makes our algorithm inefficient in general, it may be efficiently implemented in special cases and it may be operated with heuristical solutions.
We therefore believe the algorithm to be interesting beyond just providing a theoretical certificate for the existence of good incremental solutions.

Previous algorithms in the framework of incremental maximization were based on static scaling strategies that assemble optimum solutions of exponentially growing sizes in phases.
Our algorithm is the first one that fully adapts its scaling behavior to the instance and the first scaling algorithm to (deterministically) beat a factor of $\varphi + 1 \approx 2.618$ for any (nontrivial) incremental maximization variant~\cite{BernsteinDisserGrossHimburg-20, DisserWeckbecker-25,Weckbecker2023}.

We complement our result with an unconditional lower bound on the existence of good incremental solutions by providing an instance that does not admit incremental solutions that approximates the optimum to a factor better than $1.25$ across all cardinalities.

\begin{restatable}{theorem}{LowerBound}\label{thm:main-lower}\label{app:thm:L1}
    Every deterministic algorithm for incremental submodular maximization has a competitive ratio of at least~$\frac{5}{4}=1.25$.
\end{restatable}

We leave closing the gap $[1.25, 1.373]$ between our bounds as an open question.

\subparagraph*{Techniques.} 

Our algorithm is motivated by an interpretation of the classical analysis of the greedy algorithm by \ncite{NemhauserWolseyFisher1978}.
The key property of monotone submodular functions exploited in this analysis is the fact that, for every set~$S$ that does not already contain the optimum solution~$S_k^\star$, there always exists an element $e \in S_k^\star \setminus S$ with marginal return
\[
    f(S \cup \{e\}) - f(S) 
    \geq \frac{1}{k} \cdot \big(f(S_k^\star \cup S) - f(S)\big) 
    \geq \frac{1}{k} \cdot \big(f(S_k^\star) - f(S)\big).
\]
In other words, there is always an element that covers at least a $1/k$ fraction of the missing value relative to the optimum.
This is most intuitive if~$f$ is a coverage function, but holds for every monotone submodular function.
By definition, if~$S$ is the partial solution so far, the greedy algorithm includes an element not worse than the element~$e$ above.
It follows that the part of the optimum in the solution $S_i$ after step~$i$ that is not yet accounted for is bounded by
\[
    f(S_k^\star) - f(S_{i}) 
    \leq \Big(1-\frac{1}{k}\Big) \cdot (f(S_k^\star) - f(S_{i-1})) 
    \leq \ldots 
    \leq \Big(1 - \frac{1}{k}\Big)^k \cdot f(S_k^\star) 
    \leq \frac{1}{\e} \cdot f(S_k^\star),
\]
which yields the greedy competitive ratio of 
\[
  \max_k \frac{f(S_k^\star)}{f(S_i)} \leq \frac{\e}{\e - 1} \approx 1.58.
\]

We can interpret the analysis (not the actual greedy algorithm) as an algorithm that extends the current solution~$S$ by a single element towards the solution $S_k^\star \cup S$.
The key idea of our approach is to adapt this behavior by extending the solution by multiple elements towards a fixed $S_k^\star \cup S$, which is possible, since for every integer $\ell \leq |S_k^\star\setminus S|$, there exist incremental sets~$E_{1}, E_2, \ldots, E_{|S_k^\star\setminus S|} \subseteq S_k^\star\setminus S$, where $|E_\ell| = \ell$ and 
\[
    f(S\cup E_{\ell}) - f(S) \geq \frac{\ell}{k} \cdot \big(f(S_k^\star \cup S) - f(S)\big).
\]
This yields an algorithm that operates in phases, adding~$k_j$ elements towards a fixed solution in phase~$j$.
Note that the algorithm is no longer greedy since~$E$ may not include elements of the largest individual marginal increase.
In terms of its classical analysis, the greedy algorithm targets and builds towards a different solution in each step.
Targeting instead a~fixed solution for multiple steps promises to yield an improved guarantee for every cardinality of the form $k = \sum_j k_j$, since
\[
    f(S_k^\star) - f(S_k)
    \leq \prod_j \Big(1 - \frac{k_j}{k}\Big)^{k_j} \cdot f(S_k^\star) 
    \leq \Big(1 - \frac{1}{k}\Big)^k \cdot f(S_k^\star),
\]
i.e., strictly better than greedy if $k_j > 1$ for some~$j$.

While long phases (large $k_j$) achieve improved guarantees upon completion of a phase, in the worst case, we only get a linear increase in objective value during a single phase, which may be worse than greedy during most of the phase.
Previous approaches avoided this issue by completely separating the contributions in each phase and either using the contribution of the previously completed phase or of the current phase to estimate against the optimum solution for the current cardinality.
The result were scaling algorithms that use phases of lengths (at least) $\delta^i$ with a carefully balanced factor~$\delta$~\cite{BernsteinDisserGrossHimburg-20, DisserWeckbecker-25, Weckbecker2023}.
However, such strategies fail to achieve competitive ratios better than even~$\varphi + 1 \approx 2.618$ and only seem adequate for  classes of objective functions beyond the submodular regime.

In contrast, based on the described perspective on the classical greedy analysis, we are able to integrate the accounting of the different phases, which means that we are able to use the total value achieved across all phases in our estimation.
Our algorithm adapts the phase lengths dynamically so as to maximize the linear increase, i.e., the density, during a~phase.
The key insight in our analysis is that this results in the algorithm being able to afford long phases whenever~$\OPT{k} \define f(S_k^\star)$ is only slowly increasing in~$k$.
In case~$\OPT{k}$ is rapidly increasing, we are able to exploit that the greedy solution is already better than $\e/(\e-1)$-competitive. 

\subparagraph*{Related work.}

\ncite{BernsteinDisserGrossHimburg-20} initiated the study of incremental maximization with a monotone objective and presented a simple scaling algorithm that achieves a competitive ratio of~$1 + \varphi \approx 2.618$ for accountable objectives, a property that generalizes submodularity.
The best known lower bound in this regime is $2.246$ by \ncite{DisserKlimmSchewiorWeckbecker2023}.

While, subadditivity\footnote{$f\colon 2^{\U} \to \R_{\geq 0}$ is subadditive if $A,B \subseteq \U \Longrightarrow  f(A) + f(B) \geq f(A \cup B)$.}  is a more direct generalization of submodularity and a more natural ``diminishing returns'' property, it was more elusive in terms of incremental guarantees.
Recently, \ncite{DisserWeckbecker-25} were able to give the first competitive algorithm for this setting, which achieves a competitive ratio of at most $2 + \sqrt{2} \approx 3.414$.
The algorithm again uses scaling by at least a constant factor, but further increases the lengths of phases if this is beneficial for the resulting density of the phase.
Note that accountability and subadditivity are incomparable in the sense that either class captures functions not contained in the other.

The analysis of the inherently incremental greedy algorithm for cardinality constrained maximization has a long history including classical results in combinatorial optimization.
\ncite{Rado1942} observed that greedy is optimal for modular objectives, and \ncite{NemhauserWolseyFisher1978} provided the well-known classical result that greedy achieves an $\frac{\e}{\e -1}$-approximation with~$nk$ value-oracle accesses to~$f$ (meaning that the value~$f(S)$ is needed for $nk$ different sets~$S \subset \U$).
\ncite{NemhauserWolsey1978} 
and \ncite{Vondrak2013} showed that any better approximation algorithm requires an exponential number of value-oracle accesses.
For the special case of coverage functions, \ncite{Feige1998} used the PCP-theorem to establish conditional hardness of approximation to within a better factor than greedy.

The analysis of the greedy algorithm was generalized to functions of bounded submodularity ratio by \ncite{DasKempe2018}.
\ncite{BianBuhmannKrauseTschiatschek2017} proved this analysis tight and further generalized it to functions of bounded curvature.
\ncite{ElenbergKhannaDimakisNegahban2018} further extended the results of \citeauthor{DasKempe2018} and \ncite{HarshawFeldmanWardKarbasi2019} show no better approximation with polynomially many value-orcale accesses is possible for functions of bounded submodularity ratio.
Moreover, augmentability was introduced by \ncite{BernsteinDisserGrossHimburg-20} as an alternative relaxation of submodularity.
\citeauthor*{BernsteinDisserGrossHimburg-20} generalized the greedy analysis to this regime and \ncite{DisserWeckbecker-23} further parameterized the result and showed it to be tight.

Other variants of submodular maximization that have been considered include non-monotone settings~\cite{BuchbinderFeldmanNaorSchwartz2014,BuchbinderFeldmanNaorSchwartz2015,FeigeMirrokniVondrak2011}, settings with knapsack~\cite{Sviridenko2004} and matroid~\cite{CalinescuChekuriPalVondrak2011} constraints, and streaming models~\cite{BadanidiyuruMirzasoleimanKarbasiKrause2014}.


\section{The Algorithm}

Our algorithm \Alg (cf.~\Cref{alg:main_simplified}) constructs its solution $S$ incrementally, starting from the empty set. It assumes knowledge of optimum solutions $\{S^\star_k\}_{k=1}^n$ and operates in \textit{phases}. Let $k_0 = 0$. At the start of each phase $j$, we identify the solution $S_{k_j}^\star$, such that $k_j > k_{j-1}$ maximizes the value
\[ 
    \frac{\OPT{k_j} - \OPT{k_{j-1}}}{k_j}.
\]
Within the phase $j$, \Alg adds the elements of $S^\star_{k_j}$ to the solution, greedily taking the elements that increase the objective the most. The algorithm terminates once all elements have been added.

\begin{algorithm}[t]
    \DontPrintSemicolon

    $S \gets \emptyset$, $j \gets 0$, $k_0 \gets 0$\;
    \While{$S \neq \U$}{
        $j \gets j + 1$\;
        take $k_j \in \arg \max_{k > k_{j-1}} (\OPT{k} - \OPT{k_{j-1}}) / k$\; 
        \While{$S^\star_{k_j} \setminus S \neq \emptyset$}{
            take $e \in \smash{\arg \max_{e' \in S^\star_{k_j} \setminus S} f(S \cup \{e'\})}$\;
            $S \gets S \cup \{e\}$\;
        }
    }
    \caption{$\Alg(\U, f)$}
    \label{alg:main_simplified}
\end{algorithm}


\input{analysis}

\input{lower-bound}

\bibliographystyle{abbrvnat}
\bibliography{bibliography}

\appendix

\input{appendix}

\end{document}

%% file: analysis.tex
\section{Analysis}
\label{sec:analysis}

In this section, we prove the upper bound of \cref{thm:main-upper}---\Alg achieves a~competitive ratio of $\rho < 1.3729$.

\subsection{Notation}

We start by introducing variables to keep track of the phases and iterations of the algorithm explicitly (cf.~\Cref{alg:main} along with the following definitions). 
We let~$\ell$~denote the number of phases (iterations of the outer loop), and~$p_j$ the number of steps during phase $j \in [\ell]$ (iterations of the inner loop).
Observe that $1 \leq p_j \leq k_j$ and $\sum_{j=1}^\ell p_j = n = |\U|$. 

We further use $S^{(j)}_i$ to denote the state of set $S$ after iteration $i$ of phase $j$, and let $S^{(j)} \define S^{(j)}_{p_j}$ be the state of set $S$ at the end of phase $j$. 
We also denote $S^{(j)}_0 \define  S^{(j-1)}$, $S^{(j)}_{i} \define S^{(j)}_{p_j}$ for $p_j < i \leq k_j$, and
\[
    \OPTk{j} \define \OPT{k_j}.
\]
Finally, for every phase $j \in [\ell]$, we define its \emph{density} as
\[ 
    \delta_j 
    \define \frac{1}{k_j} \cdot (\OPT{k_j} - \OPT{k_{j-1}}) 
    = \frac{1}{k_j} \cdot (\OPTk{j} - \OPTk{j-1}) 
\]

It is easy to see that our algorithm ensures that densities are monotonically decreasing.

\begin{observation}
    \label{obs:densities_decrease}
    For each phase $j \in [\ell - 1]$, it holds that
    $\delta_j \geq \delta_{j+1}$.
\end{observation}

\begin{proof}
    The algorithm definition implies $k_{j-1} < k_j$. Hence, 
    \begin{align*}
        \delta_j 
        & = \max_{k>k_{j-1}} \frac{1}{k} \cdot (\OPT{k} - \OPTk{j-1}) 
        \geq \max_{k>k_{j-1}} \frac{1}{k} \cdot (\OPT{k} - \OPTk{j}) 
            && \text{(by $\OPTk{j-1} \leq \OPTk{j}$)} \\
        & \geq \max_{k>k_j} \frac{1}{k} \cdot (\OPT{k} - \OPTk{j}) 
        = \delta_{j+1}.
        && \qedhere
    \end{align*}
\end{proof}

\begin{algorithm}[t]
    \DontPrintSemicolon

    $S^{(0)} \gets \emptyset$, $j \gets 0$, $k_0 \gets 0$\;
    \While{$S^{(j)} \neq \U$}{
        $j \gets j + 1$, $S_0^{(j)} \gets S^{(j - 1)}$\;
        take $k_j \in \arg \max_{k > k_{j-1}} (1/k) \cdot (\OPT{k} - \OPT{k_{j-1}})$\; 
        $i \gets 0$\;
        \While{$S^\star_{k_j} \not \subseteq S_i^{(j)}$}{
            take $e \in \smash{\arg \max_{e' \in S^\star_{k_j} \setminus S_i^{(j)}} f(S_i^{(j)} \cup \{e'\})}$\;
            $S_{i+1}^{(j)} \gets S_{i}^{(j)} \cup \{e\}$, $i \gets i + 1$\;
        }
        $p_j \gets i$, $S^{(j)} \gets S_{p_j}^{(j)}$\;
    }
    \vspace{.2em}
    $\ell \gets j$\;
    \Return $\smash{(S^{(1)}_{1}, \dots, S^{(1)}_{p_1}, \dots, S^{(\ell)}_{1}, \dots, S^{(\ell)}_{p_\ell})}$\;
    \caption{$\Alg(\U, f)$, with variables to track phases and iterations}
    \label{alg:main}
\end{algorithm}

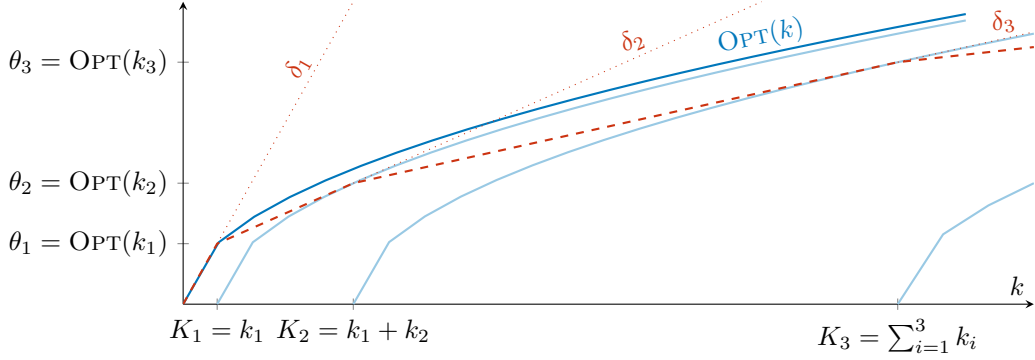
\begin{figure}
    \begin{tikzpicture}
        \begin{axis}[
            axis lines=middle,
            xlabel={$k$},
            samples=23,
            x=0.45cm,
            y=0.8cm,
            xtick={1,5,21},
            xticklabels={$K_1=k_1$,$K_2=k_1+k_2$,$K_3=\sum_{i=1}^3 k_i$},
            ytick={1,2,4},
            yticklabels={$\OPTk{1}=\OPT{k_1}$,$\OPTk{2}=\OPT{k_2}$,$\OPTk{3}=\OPT{k_3}$},
            ]
            \addplot [myblue, thick,  domain=0:23] {x^(1/2)} node[pos=0.75, above=-2, sloped] {$\OPT{k}$};
            \addplot [myblue!40!white, thick, domain=1:23, samples=22] {(x-1)^(1/2)};
            \addplot [myblue!40!white, thick, domain=5:25, samples=20] {(x-5)^(1/2)};
            \addplot [myblue!40!white, thick, domain=21:25, samples=4] {(x-21)^(1/2)};
            
            \addplot[thick, dashed,	myred]
            coordinates {
                (0,0)
                (1,1)     
                (5,2)     
                (21,4)    
                (25,4.25) 
            };
            \addplot[myred, dotted] coordinates {(1,1) (5,5)} node[pos=0.7,sloped,above=-1.5] {$\delta_1$}; 
            \addplot[myred, dotted] coordinates {(5,2) (17,5)} node[pos=0.7,sloped,above=-1.5] {$\delta_2$}; 
            \addplot[myred, dotted] coordinates {(21,4) (25,4.5)} node[pos=0.8,sloped,above=-1.5] {$\delta_3$}; 
        \end{axis}
    \end{tikzpicture}
    \caption{The phases of the algorithm: The dashed red curve describes the lower bound on the value of the algorithm given by \Cref{lem:algocurve}. In phase $j$, the cardinality $k_j$ is chosen such that the red curve is tangential to the optimum curve starting in the beginning of phase $j$ (solid blue plots) which corresponds to maximizing the density $\delta_j$ describing the slope of the lower bound.}
    \label{fig:phases}
\end{figure}


\subsection{Basic Properties}
\label{sec:basic_properties}

We now prove a few properties that relate the gains of \Alg and \OPTAlg. 

First, observe that at the end of phase $j$, \Alg has accumulated all elements of the optimum solution of size $k_j$. This happens after it collects \smash{$\sum_{i=1}^j p_i \leq \sum_{i=1}^j k_i$} elements. 

\begin{lemma}
    \label{lem:1}
    For each phase $j \in [\ell]$, we have $f(S^{(j)}) \geq \OPTk{j}$.
\end{lemma}

\begin{proof}
    At the end of the $j^{th}$ phase, we have added each element in $S_{k_j}^\star$ to the solution~$S$. Thus, $S^{(j)} \supseteq S_{k_j}^\star$, and so $f(S^{(j)}) \geq f(S_{k_j}^\star) = \OPTk{j}$ by the monotonicity of~$f$.
\end{proof}

Fix any phase $j$ of the algorithm and recall that the number of iterations in the phase is $p_j \leq k_j$. 
\Cref{lem:1} lower-bounds \Alg's gain by $\OPTk{j-1}$ and $\OPTk{j}$ at the beginning and end of phase $j$, respectively.
Linear interpolation of these two bounds evaluated after the $i^{th}$ step of phase~$j$ yields a value of 
$\OPTk{j-1} + (i/p_j) \cdot (\OPTk{j} - \OPTk{j-1}) 
    \geq \OPTk{j-1} + (i/k_j) \cdot (\OPTk{j} - \OPTk{j-1})
    = \OPTk{j-1} + i \cdot \delta_j$. 

In the next lemma, we show that the submodularity of function $f$ combined with greediness of \Alg guarantee at least this ``interpolated gain'' (cf.~\Cref{fig:phases}). 
In the proof, we will use the following claim that follows directly by submodularity of $f$ (see \Cref{sec:missing_proofs} for a proof).

\begin{restatable}{claim}{submodularitygivesatleastaverage}
    \label{cla:submodularity_gives_at_least_average}
    For every two subsets $A \subsetneq B \subseteq \U$ there exists an element $e \in B \setminus A$, such that 
    $ f(A \cup \{e\}) - f(A) \geq (f(B) - f(A)) / (|B| - |A|)$.
\end{restatable}

\begin{lemma}
    \label{lem:algocurve}
    For each phase $j \in [\ell]$ and $i \in [k_j]$, we have $f(S^{(j)}_i) \geq \OPTk{j-1} + i \cdot \delta_j$.
\end{lemma}

\begin{proof}
    For $i \geq p_j$, we have $S^{(j)}_{i}=S^{(j)}$. By \Cref{lem:1} and by the definition of density $\delta_j$, it holds that $f(S^{(j)}_{i}) = f(S^{(j)}) \geq \OPTk{j} = \OPTk{j-1} + k_j \cdot \delta_j \geq \OPTk{j-1} + i \cdot \delta_j$.
    Hence, in the remaining part of the proof, we focus on case $i \in [p_j]$.
    
    Fix $t \in [p_j]$. By \Cref{cla:submodularity_gives_at_least_average} applied to sets $S_{t-1}^{(j)}$ and $S^{(j)}$, there exists an element $e_t \in S^{(j)} \setminus S_{t - 1}^{(j)}$ such that
    \begin{equation}\label{eq:marginal_increase}
        \frac{f(S^{(j)}) - f(S_{t - 1}^{(j)})}{p_j - t + 1}
        \leq f(S_{t - 1}^{(j)} \cup \{e_t\}) - f(S_{t - 1}^{(j)}) 
        \leq f(S_{t}^{(j)}) - f(S_{t - 1}^{(j)}).        
    \end{equation}
    Let $d_t \define f(S^{(j)}) - f(S^{(j)}_t)$ and $g_t \define f(S^{(j)}_t) - f(S^{(j)}_{t-1})$. By \eqref{eq:marginal_increase}, it holds that $d_t = d_{t - 1} - g_t$ and $g_t \geq d_{t - 1} / (p_j - t + 1)$. Thus, $d_t \leq d_{t - 1} \cdot (1 - 1/(p_j - t + 1))$, and we have
    \begin{equation}
        \label{eq:di_bound}
        d_i 
        \leq d_0 \cdot \prod_{t = 1}^i \left( 1 - \frac{1}{p_j - t + 1} \right) 
        = d_0 \cdot \prod_{t = 1}^i \frac{p_j - t}{p_j - t + 1} 
        = d_0 \cdot \frac{p_j - i}{p_j} 
        \leq d_0 \cdot \left( 1  - \frac{i}{k_j} \right).
    \end{equation}
    By the definition of $d_i$ and \eqref{eq:di_bound}, 
    \[ 
        f(S^{(j)}) - f(S^{(j)}_i) = d_i \leq d_0 \cdot \left( 1  - \frac{i}{k_j} \right) = \left( f(S^{(j)}) - f(S^{(j-1)}) \right) \cdot \left( 1  - \frac{i}{k_j} \right),
    \]
    which can be rearranged to
    \begin{align*}
        f(S^{(j)}_i) 
        & \geq \frac{i}{k_j} \cdot f(S^{(j)}) + \left(1-\frac{i}{k_j}\right) \cdot f(S^{(j-1)}) \\
        & \geq \frac{i}{k_j} \cdot \OPTk{j} + \left(1-\frac{i}{k_j}\right) \cdot \OPTk{j-1} 
            && \text{(by \Cref{lem:1} and $i \leq k_j$)} \\
        & = \OPTk{j-1} + i \cdot \delta_j.
        && \qedhere
    \end{align*} 
    
\end{proof}

Finally, we argue that the choice of the density made by \Alg can also be used to upper-bound the gain of \OPTAlg between phases. 

\begin{lemma}
    \label{lem:bound_opt_with_density}
    For each phase $t \in [\ell]$ and $k \in [n]$, we have $\OPT{k} \leq \OPTk{t} + (k-k_t) \cdot \delta_t$.
\end{lemma}
\begin{proof}
    We first argue that $(\OPT{k}-\OPTk{t-1})/k \leq (\OPTk{t}-\OPTk{t-1})/k_t$. This is immediately true if $k > k_{t-1}$ by the choice of $k_t$ made by the algorithm. On the other hand, when $k \leq k_{t-1}$, the left hand side is at most $0$ and the right hand side at least $0$. 
    
    Next, the relation $(\OPT{k}-\OPTk{t-1})/k \leq (\OPTk{t}-\OPTk{t-1})/k_{j}$ can be rearranged to
    \begin{align*}
        \OPT{k} 
        & \leq \OPTk{t-1} + \frac{k}{k_t} \cdot (\OPTk{t} - \OPTk{t-1}) 
        = \OPTk{t} + \left(\frac{k_t - k}{k_t} \right) \cdot (\OPTk{t-1} - \OPTk{t}) \\
        & = \OPTk{t} + (k_t - k) \cdot (-\delta_t).
        \qedhere
    \end{align*}
\end{proof}

\tracinglostchars=0  

\begin{figure}
    \centering
    \begin{tikzpicture}
        \begin{axis}[
            axis lines=middle,
            xlabel={$k/k_j$},
            samples=96,	
            x=1.5cm,
            y=1.5cm,
            xtick={1,{21/16}}, 
            xticklabels={$1$,$c_j$},
            ytick={0.5,1},
            yticklabels={$\OPTk{j-1} / \OPTk{j}$,$1$},
            ]
            \addplot [myblue, thick,  domain=0:6] {x^(1/2)}; 
            \addplot[thick, dashed,	myred]
            coordinates {
                (0,0)
                ({1/16},{1/4})     
                ({5/16},{1/2})     
                ({21/16},1)    
                ({85/16},2)    
                (6,2.0859375)       
            };
            \addplot[dotted,black] coordinates{(0,{1/2}) (1,1) (4,2.5)} node[pos=0.8, above=-2, sloped] {$q_{j-1}$};
            \addplot[dotted,black] coordinates{(0,1) (4,2) (6,2.5)} node[pos=0.1, above=-2, sloped] {$q_{j}$};
        \end{axis}
        \node[myblue] at (9.6,3.65) 
            {$\frac{\OPT{k}}{\OPTk{j}}$};
    \end{tikzpicture}
    \caption{
    When analyzing the relation between~$c_j$ and the normalized density~$q_j$ in phase $j$ in \Cref{lem:relation_of_slack_and_slope}, we scale all values down by $\OPTk{j}$ and all cardinalities by $k_j$. 
    The solid blue plot depicts the rescaled optimum curve while the dashed red plot shows the rescaled lower bound. Note that the dotted slopes~$q_{j-1}$ and~$q_j$ are exactly the slopes of the rescaled lower bound before and after $c_j$.}
    \label{fig:definitions_in_analysis}
\end{figure}
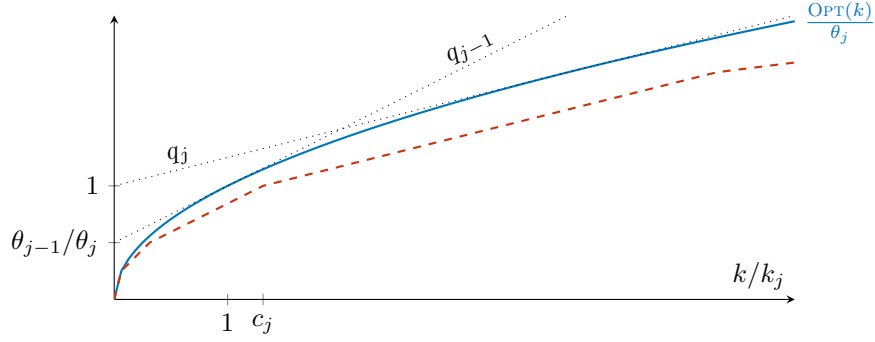

\tracinglostchars=1


\subsection{Inductive Analysis}

In the analysis of phase $j$, we exploit the choice of cardinality $k_j$ by \Alg together with the value~$\OPTk{j}$ of the optimal solution at this cardinality. 
To measure the progress of the algorithm in phase $j$ it can be convenient to scale all cardinalities by $k_j$ and all gains by~$\OPTk{j}$. 
Accordingly, we define \emph{normalized density} as
\[ 
    q_j \define \frac{k_{j}}{\OPTk{j}} \cdot \delta_j 
    = \frac{\OPTk{j}-\OPTk{j-1}}{\OPTk{j}}
    = 1-\frac{\OPTk{j-1}}{\OPTk{j}}.
\]
We also define 
\[
    \textstyle K_j \define \sum_{i=1}^j k_i 
    \qquad \text{and} \qquad
    c_j = K_j / k_j.
\]
See \Cref{fig:definitions_in_analysis} for an illustration of these definitions. 

By \Cref{lem:1}, our algorithm reaches (or exceeds) the value $\OPTk{j}$ after collecting $\sum_{i=1}^j p_i \leq \sum_{i=1}^j k_i = K_j$ elements, while this value is reached by \OPTAlg after collecting $k_j$ elements.
Hence, $c_j = K_j / k_j$ measures the multiplicative ``cardinality overhead'' our algorithm incurs in comparison to \OPTAlg. 
We now show that this overhead can be controlled, i.e., that we can link its values between consecutive phases, and the change of the overhead can be attributed to the normalized density.

More formally, we define the function $h \colon (0,1] \to \R$ as
\[
    h(q) \define \frac{156+1543 \cdot q-550 \cdot q^2-149 \cdot q^3}{1000 \cdot q}.
\]
We use some technical properties of function $h$, listed in the following claim (see~\Cref{app:conditions_on_h} for a proof). 

We do not use the actual definition of function $h$ in the remaining part of the analysis, but only the fact that it satisfies the properties of \Cref{cla:properties_of_h} for a given value of~$\rho$. 
Any other function (potentially giving a different $\rho$) satisfying these constraints can be applied in our analysis, yielding~$\rho$ as a bound on the competitive ratio of \Alg.

\begin{restatable}{claim}{propertiesofh}
    \label{cla:properties_of_h}
    Let $\rho = 1.3729$. Function $h$ satisfies the following properties:
    \setcounter{property}{0}%
    \begin{enumerate}
        \item \refstepcounter{property}\label{prop:h_decreasing} 
            Function $h(q)$ is decreasing in $q$ for $q > 0$ and $h(1) = 1$.
        \item \refstepcounter{property}\label{prop:first_condition_on_h} 
            For all $x> 0$ and $q \in (0,1]$, we have
            \begin{align*}
                1+\frac{h(q)}{1+x}\leq h\left(\frac{x q}{1+x q}\right).
            \end{align*}
        \item \refstepcounter{property}\label{prop:second_condition_on_h}
            For all $q \in (0,1]$ we have $1+(h(q)-1) \cdot q \leq \rho$.
        \item \refstepcounter{property}\label{prop:third_condition_on_h}
            For all $x \geq 0$ and $q \in (0,1]$, we have
            \begin{align*}
                1-\frac{h(q) \cdot x q}{(1+ xq)(1+x)} \geq \rho^{-1}.
            \end{align*}
    \end{enumerate}
\end{restatable}

This allows us to relate the cardinality overhead to the normalized density of a phase via~$h$.

\begin{lemma}
    \label{lem:relation_of_slack_and_slope}
    In each phase $j \in [\ell]$, we have $c_j \leq h(q_{j})$.
\end{lemma}
\begin{proof}
    We prove the lemma by induction on $j$. For the inductive basis ($j=1$), we observe that $c_1=1=h(1)\leq h(q_1)$. 
    
    Now, we assume the inductive hypothesis hold for phase $j-1$ and we show it for phase $j > 1$. Let $x = (k_j-k_{j-1}) / k_{j-1}$. Note that $x > 0$ as $k_j > k_{j-1}$ by the choice of the algorithm. By \Cref{lem:bound_opt_with_density} applied with $k = k_j$ and $t = j-1$, we obtain
    \[
        \OPTk{j}\leq \OPTk{j-1} + (k_j-k_{j-1}) \cdot \delta_{j-1}
        = \OPTk{j-1} + (k_j-k_{j-1}) \cdot \frac{\OPTk{j-1}}{k_{j-1}} \cdot q_{j-1}
        = (1+x \cdot q_{j-1})\cdot \OPTk{j-1}.
    \]
    Using the bound above, we can bound the normalized density as 
    \begin{align}
        \label{eq:bounding_slope}
        q_{j} = 1-\frac{\OPTk{j-1}}{\OPTk{j}} 
        \leq 1-\frac{\OPTk{j-1}}{(1+x \cdot q_{j-1}) \cdot \OPTk{j-1}} =\frac{x \cdot q_{j-1}}{1+x \cdot q_{j-1}}.
    \end{align}
    We conclude that
    \begin{align*}
        c_{j} 
        & = \frac{\sum_{i=0}^{j} k_i}{k_{j}}
            = 1+c_{j-1} \cdot \frac{k_{j-1}}{k_j}
            = 1+\frac{c_{j-1}}{1+x} \\
        & \leq 1+\frac{h(q_{j-1})}{1+x} 
            && \text{(by the inductive hypothesis)} \\
        & \leq h\left(\frac{x \cdot q_{j-1}}{1+x \cdot q_{j-1}}\right) 
            && \text{(by \Cref{prop:first_condition_on_h} of \Cref{cla:properties_of_h})} \\
        & \leq h(q_{j}),
            && \text{(by \eqref{eq:bounding_slope} and \Cref{prop:h_decreasing} of \Cref{cla:properties_of_h})}
    \end{align*}
    which proves the inductive claim, and thus concludes the proof.
\end{proof}

\subsection{Bounds on the Optimum}

Recall that once the algorithm reaches the end of phase $j$, it has collected $\sum_{i=1}^j p_j \leq \sum_{i=1}^j k_i = K_j$ elements. 
Hence, our goal now is to upper-bound the gain of \OPTAlg once it has collected the same number of elements, as a benchmark for the gain of \Alg at that point. 

For every phase $j \in [\ell-1]$, we define 
\[
    k^*_j = k_j \cdot \frac{\delta_{j}}{\delta_{j}-\delta_{j+1}}.
\]
Note that $k^*_j$ is not necessarily an integer and may not correspond directly to a cardinality.

\begin{lemma}
    \label{lem:kstar_properties}
    For each phase $j \in [\ell-1]$, it holds that $k_j \leq k^*_j$
    and $(k^*_j - k_j) \cdot \delta_j = k^*_j \cdot \delta_{j+1}$.
\end{lemma}

\begin{proof}
    By \Cref{obs:densities_decrease}, $\delta_j \geq \delta_{j-1}$, and thus 
    $k^*_j \geq k_j$. The second property follows as
    $(k^*_j - k_j) \cdot \delta_j - k^*_j \cdot \delta_{j+1} = k^*_j \cdot (\delta_j - \delta_{j+1}) - k_j \cdot \delta_j = 0$.
\end{proof}

The intuition behind $k^*_j$ is that it is the cardinality where the two bounds on $\OPT{k}$ given by \Cref{lem:bound_opt_with_density} applied with $t = j$ and $t = j+1$ intersect. To see this, recall that \Cref{lem:bound_opt_with_density} applied with $t = j$ yields
\[
    \OPT{k} \leq \OPTk{j}+ (k-k_{j})\cdot \delta_{j}.
\]
On the other hand, \Cref{lem:bound_opt_with_density} applied with $t = j+1$ yields 
\begin{align*}
    \OPT{k} &\leq \OPTk{j+1}+ (k-k_{j+1})\cdot \delta_{j+1} \\
    &= \OPTk{j}+ k_{j+1} \cdot \delta_{j+1}+ (k-k_{j+1})\cdot \delta_{j+1} \\
    &=\OPTk{j}+ k \cdot \delta_{j+1}.
\end{align*}

We conclude that the above bounds coincide at~$k^{\star}_j$ and use \Cref{lem:kstar_properties} to relate these bounds to~$\rho$.

\begin{lemma}
    \label{lem:good_approx_for_opt_bound_intersection}
    Fix $\rho$ as in \Cref{cla:properties_of_h}.
    For $j\in [\ell]$, we have 
    \[ 
        \OPTk{j}+ (k^*_j-k_{j})\cdot \delta_{j} = \OPTk{j}+ k^*_j \cdot \delta_{j+1} \leq \rho \cdot (\OPTk{j}+ (k^*_j-K_j) \cdot \delta_{j+1}).
    \]
\end{lemma}
\begin{proof}
    The first inequality is due to the definition of $k^*_j$ and is restated in \Cref{lem:kstar_properties}. Thus, it suffices to bound $\OPTk{j}+ (k^*_j-k_{j})\cdot \delta_{j}$. Let $x = (k^*_j-k_j)/k_j$ be the normalized distance between $k^*_j$ and $k_j$. Then, 
    \[
        \frac{k^*_j-k_j}{k^*_j} = \frac{x}{1+x} 
            \qquad \text{and} \qquad
        \frac{k^*_j}{k_j} = 1+x.
    \]
    Moreover,
    \begin{align}
        \label{eq:normalized_upper_bound_on_opt}
        \OPTk{j} + (k^*_j-k_j) \cdot \delta_{j} = \OPTk{j} \cdot \left( 1 + \frac{k^*_j-k_j}{k_j} \cdot q_j\right) = \OPTk{j} \cdot (1+xq_j).
    \end{align}
    On the other hand, by \Cref{lem:kstar_properties},
    we have $\delta_{j+1} = \delta_j \cdot (k^*_j-k_{j})/k^*_j = q_j \cdot (\OPTk{j}/k_j) \cdot (k^*_j-k_{j})/k^*_j$, and thus
    \begin{align}
        \nonumber
        \OPTk{j} + (k^*_j-K_j) \cdot \delta_{j+1} 
        & = \OPTk{j} \cdot \left( 
            1 + \frac{k^*_j-K_j}{k_j} \cdot    \frac{k^*_j-k_{j}}{k^*_j} \cdot q_j 
            \right) \\
        \nonumber
        & = \OPTk{j} \cdot \left( 
            1 + (1+x-c_j) \cdot \frac{x}{1+x} \cdot q_j 
            \right) \\
        \label{eq:upper_bound_on_opt_at_k_star}
        & = \OPTk{j} \cdot \left( 
            1+x q_j - \frac{c_j \cdot x q_j}{1+x}
            \right).
    \end{align}
    Combining \eqref{eq:normalized_upper_bound_on_opt} with \eqref{eq:upper_bound_on_opt_at_k_star}, we obtain 
    \begin{align*}
        \frac{\OPTk{j}+(k^*_j-K_j) \cdot \delta_{j+1}}{\OPTk{j} + (k^*_j-k_{j})\cdot \delta_{j}} 
        & = 1-\frac{c_{j} \cdot xq_{j}}{(1+ xq_{j})(1+x)}\\
        & \geq 1-\frac{h(q_{j}) \cdot xq_{j}}{(1+ xq_{j})(1+x)} 
            && \text{(by \Cref{lem:relation_of_slack_and_slope})}\\
        & \geq 1/\rho.
            && \text{(by \Cref{prop:third_condition_on_h} of \Cref{cla:properties_of_h})}
        \qedhere
    \end{align*}
\end{proof}

Now, we extend the good approximation result 
to all cardinalities within a~phase. For this, we use linearity of the upper bound on the optimum value given in \Cref{lem:bound_opt_with_density}.

\begin{lemma}
    \label{lem:good_approx_helper}
    Fix $t \in \set{0, \dots, \ell-1}$ and any real $y$ such that $K_t \leq y \leq K_{t+1}$. Then, for every phase $j \in \set{0, \dots, \ell-1}$, it holds that
    $
        \OPTk{t}+\left(y- K_t\right) \cdot \delta_{t+1} \leq  \OPTk{j}+\left(y- K_j\right) \cdot \delta_{j+1}
    $.
\end{lemma}

\begin{proof}
    For any $r \in \set{0, \dots, \ell-1}$, we define $A_r \define \OPTk{r}+ (y - K_r) \cdot \delta_{r+1}$. We will show that $A_t \leq A_j$ for any $j$. First, we observe that for each $r \geq 1$ it holds that
    \begin{align*}
        A_r - A_{r-1} 
        & = (\OPTk{r}+ (y - K_r) \cdot \delta_{r+1}) - (\OPTk{r-1} + (y-K_{r-1}) \cdot \delta_{r}) \\
        &=(\OPTk{r} + (y - K_{r} ) \cdot \delta_{r+1}) - 
        (\OPTk{r-1} + k_{r} \cdot \delta_{r}+ (y - K_{r}) \cdot \delta_{r}) \\
        &= (y - K_r) \cdot (\delta_{r+1} - \delta_r).
    \end{align*}
    By \Cref{obs:densities_decrease}, it holds that $\delta_r \geq \delta_{r+1}$. Hence, $A_0 \geq A_1 \geq \dots \geq A_{t}$ and $A_t \leq A_{t+1} \leq \dots \leq A_{\ell}$, which concludes the proof.
\end{proof}

\begin{lemma}
    \label{lem:good_approx}
    Fix $\rho$ as in \Cref{cla:properties_of_h}.
    For all $k \in [n]$ and $j \in \set{0, \dots, \ell-1}$, it holds that
    \[
        \OPT{k} \leq \rho \cdot (\OPTk{j}+ \left(k- K_j \right) \cdot \delta_{j+1}).
    \]
\end{lemma}

\begin{proof}
    Throughout the proof, we fix $j \in \set{0, \dots, \ell}$. We prove the lemma by induction on~$k$. 
    
    First, observe that for $k=0$ \Cref{lem:good_approx_helper}, applied with $t = 0$ and $y = 0$, yields $\OPTk{j}+ \left(0- K_j \right) \cdot \delta_{j+1} \geq \OPTk{0}+ \left(0- K_0 \right) \cdot \delta_{1} = \OPTk{0} = \OPT{0}$, and thus the induction basis holds.

    Assuming the lemma statement holds for all $k \leq k_{t}$ for some $t \in \set{0, \dots, \ell-1}$, we show that it holds for all $k'\in \set{ k_t+1, k_t+2, \dots, k_{t+1} }$. 
    We have
    \begin{align}
        \nonumber
        \OPTk{t} + K_{t-1} \cdot \delta_t 
        & = \OPTk{t} + K_{t-1} \cdot (q_t/k_t) \cdot \OPTk{t} 
            && \text{(by definition of $q_j$)} \\
        \nonumber
        & = (1+(c_{t}-1) \cdot q_t)  \cdot \OPTk{t} 
            && \text{(by definition of $c_j$)} \\
        \nonumber
        & \leq (1+(h(q_t)-1) \cdot q_t) \cdot \OPTk{t} 
            && \text{(by \Cref{lem:relation_of_slack_and_slope})} \\
        \nonumber
        & \leq \rho \cdot \OPTk{t}.
            && \text{(by \Cref{prop:second_condition_on_h} of \Cref{cla:properties_of_h})} \\
        \nonumber
        & = \rho \cdot (\OPTk{t}+ (K_t - K_t) \cdot \delta_{t+1}) \\
        \label{eq:good_approx_at_phase_end}
        & \leq \rho \cdot (\OPTk{j} + (K_t - K_j) \cdot \delta_{j+1}).
            && \text{(by \Cref{lem:good_approx_helper} with $y = K_t$)} 
    \end{align}
    
    We now distinguish cases depending on whether or not $k' \leq K_t$.

    \begin{description}
        \item[Case 1.] Assume $k' \leq K_t$. By induction, the statement holds for $k=k_{t}$, resulting in
    \begin{align}
        \label{eq:good_approx_opt}
        \OPTk{t} \leq \rho \cdot (\OPTk{j}+ (k_{t}-K_j) \cdot \delta_{j+1}).
    \end{align}
    We fix $\alpha \in [0,1]$ such that $k'= (1-\alpha) \cdot k_t + \alpha \cdot K_t = k_t + \alpha \cdot K_{t-1}$. Then,
    \begin{align*}
        \OPT{k'} 
        & \leq \OPTk{t} + (k'-k_t) \cdot \delta_{t} 
            && \text{(by \Cref{lem:bound_opt_with_density})} \\
        & = (1-\alpha) \cdot \OPTk{t} + \alpha \cdot ( \OPTk{t} + K_{t-1} \cdot \delta_{t}) \\
        & \leq (1-\alpha) \cdot \rho \cdot (\OPTk{j}+ (k_{t}-K_j ) \cdot \delta_{j+1}) \\
        & \qquad + \alpha \cdot \rho \cdot (\OPTk{j}+ (K_t-K_j) \cdot \delta_{j+1})
            && \text{(by \eqref{eq:good_approx_opt} and \eqref{eq:good_approx_at_phase_end})} \\
        & = \rho \cdot (\OPTk{j}+ ( k'-K_j) \cdot \delta_{j+1}).
    \end{align*}

    \item[Case 2.] Assume $K_t < k' \leq k_{t+1}$. (This case holds vacuously if $K_t \geq k_{t+1}$). 
    Note that $K_t < k' \leq k_{t+1} \leq K_{t+1}$.

    \item[Case 2.1.] Consider the sub-case $k'\leq k^*_t$. 
    In this case, we first observe that 
    \begin{align}
    \nonumber
    \OPTk{t} + (k^*_t-k_t) \cdot \delta_{t}
        & \leq \rho \cdot (\OPTk{t}+ (k^*_t-K_t) \cdot \delta_{t+1}) \\
    \label{eq:good_approx_at_k_star}
        & \leq \rho \cdot (\OPTk{j}+ (k^*_t-K_j) \cdot \delta_{j+1}),
    \end{align}
    where the first inequality is due to \Cref{lem:good_approx_for_opt_bound_intersection} and the second one follows from \Cref{lem:good_approx_helper} applied with $y = k^*_t$.
    Next, we fix $\alpha \in [0,1]$ such that $k'= (1-\alpha) \cdot K_t + \alpha \cdot k^*_t$. Then,
    \begin{align*}
        \OPT{k'} 
        & \leq \OPTk{t} + (k'-k_t) \cdot \delta_{t}
            && \text{(by \Cref{lem:bound_opt_with_density})} \\
        & = (1-\alpha) \cdot (\OPTk{t} + K_{t-1} \cdot \delta_{t})+ \alpha \cdot (\OPTk{t} + (k^*_t-k_t) \cdot \delta_{t}) \\
        & \leq (1-\alpha) \cdot \rho \cdot (\OPTk{j}+ (K_t-K_j) \cdot \delta_{j+1}) \\
            & \qquad + \alpha \cdot \rho \cdot (\OPTk{j}+ (k^*_t-K_j) \cdot \delta_{j+1})
            && \text{(by \eqref{eq:good_approx_at_phase_end} and \eqref{eq:good_approx_at_k_star}} \\
        & = \rho \cdot (\OPTk{j}+ (k' - K_j) \cdot \delta_{j+1}).
    \end{align*}
    
    \item[Case 2.2.] Consider the sub-case $k^*_t < k'$. 
    Then,
    \begin{align*}
        \OPT{k'} 
        & \leq \OPTk{t+1}+ (k'-k_{t+1}) \cdot \delta_{t+1} && \text{(by \Cref{lem:bound_opt_with_density})} \\
        & = \OPTk{t}+ k_{t+1} \cdot \delta_{t+1}+ (k'-k_{t+1}) \cdot \delta_{t+1} \\
        & = \OPTk{t}+ k' \cdot \delta_{t+1} \\
        & = \OPTk{t}+ k^*_t \cdot \delta_{t+1} + (k'-k^*_t) \cdot \delta_{t+1}\\
        & \leq \rho \cdot (\OPTk{t}+(k^*_t-K_t) \cdot \delta_{t+1}) + (k'-k^*_t) \cdot \delta_{t+1}
            && \text{(by \Cref{lem:good_approx_for_opt_bound_intersection})} \\
        & \leq \rho \cdot (\OPTk{t}+(k'-K_t) \cdot \delta_{t+1})
            && \text{(as $k' > k^*_t$ and $\rho \geq 1$)} \\ 
        & \leq \rho \cdot (\OPTk{j}+(k'-K_j)\cdot \delta_{j+1}) 
            && \text{(by \Cref{lem:good_approx_helper} with $y = k'$).} 
    \end{align*}
    \end{description}
    In all cases, we showed that $\OPT{k'} \leq \rho \cdot (\OPTk{j}+ ( k'-K_j) \cdot \delta_{j+1})$, which concludes the induction step.

    Finally, for $k'\in [n]$ with $k'\geq k_{\ell}$ we use that the statement holds for $k=k_{\ell}$ to infer
    \begin{align*}
        \OPT{k'}
        & \leq \OPTk{\ell} + (k'-k_{\ell})\cdot \delta_{\ell}
            && \text{(by \Cref{lem:bound_opt_with_density})} \\
        & \leq \rho \cdot (\OPTk{j}+(k_{\ell}-K_j)\cdot \delta_{j+1})+ (k'-k_{\ell})\cdot \delta_{\ell}\\
        & \leq \rho \cdot (\OPTk{j}+(k_{\ell}-K_j)\cdot \delta_{j+1}+ (k'-k_{\ell})\cdot \delta_{\ell})
            && \text{(as $k' \geq k_{\ell}$ and $\rho \geq 1$)} \\ 
        & \leq \rho \cdot (\OPTk{j}+(k'-K_j)\cdot \delta_{j+1}),
            && \text{(as $k' \geq k_{\ell}$ and $\delta_{\ell} \leq \delta_{j+1}$)}       
    \end{align*}
    which concludes the proof.
\end{proof}

By combining \Cref{lem:algocurve} and \Cref{lem:good_approx}, we obtain our main result.

\mainupper*

\begin{proof}
    We fix $\rho$ as in \Cref{cla:properties_of_h}, i.e., $\rho=1.3729$. 
    We fix an integer $k \in [n]$ and we will show that 
    $f(S_k) \geq \OPT{k} / \rho$.
    
    If $k \leq k_1$, then $f(S_k) = f(S_k^{(1)}) \geq k \cdot \delta_1 = \OPTk{1} + (k-k_1) \cdot \delta_1 \geq \OPT{k}$, where the first and the second inequality follow by \Cref{lem:algocurve} and \Cref{lem:bound_opt_with_density}, respectively.
    
    If $k \geq k_1$, we choose largest $t\in [\ell]$ satisfying $K_t \leq k$. Then,
    \begin{align*}
        \OPT{k} 
        & \leq \rho \cdot (\OPTk{t}+ (k- K_t) \cdot  \delta_{t+1}) 
        && \text{(by \Cref{lem:good_approx} with $j=t$)} \\
        & \leq \rho \cdot f(S^{(j)}_{i}) 
        && \text{(by \Cref{lem:algocurve} with $j=t+1$ and $i=k-K_t$)} \\
        & = \rho \cdot f(S_{k'}) 
        && \text{(with $\textstyle k'=\sum_{r=1}^{j-1} p_r + i$)} \\
        & \leq \rho \cdot f(S_k).
        && \text{(as $k'\leq K_t + i = k$)} 
        && \qedhere
        \end{align*}
\end{proof}

%% file: lower-bound.tex
\section{Lower Bounds}

In this section, we provide the full description of our lower bounds. In \Cref{sec:genLower}, we provide a lower bound on the competitive ratio of any algorithm, proving the following theorem (restated).

\LowerBound*

In \Cref{sec:algLower}, we give an improved lower bound against our algorithm, showing that the analysis of \Alg of \cref{thm:main-upper} is tight.

\begin{restatable}{theorem}{theoremalgolower}
    \label{app:thm:L2}
     \Cref{alg:main} has a competitive ratio of at least $1.3724$.
\end{restatable}

\subsection{A Hard Instance for Incremental Algorithms}\label{sec:genLower}

Consider a set $X \define \{x_{i,j} \mid i \in [5], j \in [4]\}$ of $20$ elements, and the subsets $R_i = \{ x_{i,j} \mid j \in [4]\}$ and $C_j = \{ x_{i,j} \mid i \in [5]\}$. Let $f$ be the coverage function defined over the collection of subsets $\U = \{R_1, \dots, R_5, C_1, C_2\}$. In other words, for every $S \subseteq \U$, we have $f(S) = \left|\bigcup_{E \in S} E \right|$.

The following lemma shows that \emph{any} incremental solution for the function $f$ has a~competitive ratio of at least $1.25$, proving \Cref{app:thm:L1}.
Consider \Cref{fig:LB} alongside its proof.

\begin{figure}
\begin{tikzpicture}
  \definecolor{mygreen}{HTML}{009988}
  \definecolor{myorange}{HTML}{EE7733}

  \def\PicoScale{1.25} 
  \def\DotScale{1.8}  

  \def\picbase{0.78mm}

  \tikzset{
    pics/pico/.style args={r #1 c #2}{code={
      \pgfmathtruncatemacro{\r}{#1}
      \pgfmathtruncatemacro{\c}{#2}

      \pgfmathsetmacro{\lw}{0.75}      
      \pgfmathsetmacro{\rc}{1.50}      
      \pgfmathsetmacro{\dr}{0.115*\DotScale} 

      \foreach \i in {0,...,4} {%
        \ifnum\i<\r
          \draw[mygreen, line width=\lw pt, rounded corners=\rc pt] (0,\i) rectangle (4,\i+1);
        \fi
      }

      \foreach \j in {0,...,3} {%
        \ifnum\j<\c
          \draw[myorange, line width=\lw pt, rounded corners=\rc pt] (\j,0) rectangle (\j+1,5);
        \fi
      }

      \foreach \i in {0,...,4} {%
        \foreach \j in {0,...,3} {%
          \pgfmathtruncatemacro{\covered}{(\i<\r) || (\j<\c) ? 1 : 0}
          \ifnum\covered=1
            \fill (\j+0.5,\i+0.5) circle (\dr);
          \fi
        }%
      }%
    }},
  }

  \begin{axis}[width=13.2cm,height=7.6cm,xmin=0,xmax=7,ymin=0,ymax=22,xtick={0,1,2,3,4,5,6,7},ytick={0,5,10,15,20},axis lines=left,axis line style={-stealth},tick style={black},grid=both,major grid style={draw=black,opacity=0.12},minor grid style={draw=black,opacity=0.07},xlabel={$k$},ylabel={$f$},enlargelimits=false,clip=false]

    \addplot[very thick,dashed,myorange,mark=*,mark size=2pt,mark options={draw=myorange,fill=myorange,line width=0pt}]
      coordinates {(0,0) (1,5) (2,10) (3,12) (4,14) (5,16) (6,18) (7,20)};
    \addplot[very thick,mygreen,mark=*,mark size=2pt,mark options={draw=mygreen,fill=mygreen,line width=0pt}]
      coordinates {(0,0) (1,4) (2,8) (3,12) (4,16) (5,20) (6,20) (7,20)};

    \node[scale=0.75] at (axis cs:2,9)  {$\frac{10}{8}$};
    \node             at (axis cs:5,18) {$\frac{20}{16}$};

    \pgfmathsetmacro{\picunit}{\PicoScale*\picbase}
    \newcommand{\Pico}[2]{\tikz[x=\picunit,y=\picunit,baseline=-0.3ex]{\pic{pico={r #1 c #2}};}}

    \def\UP{4pt}
    \def\DOWN{-4pt}

    \node[anchor=south,yshift=\UP]   at (axis cs:1,5)  {\Pico{0}{1}};
    \node[anchor=south,yshift=\UP]   at (axis cs:2,10) {\Pico{0}{2}};
    \node[anchor=north,yshift=\DOWN] at (axis cs:3,12) {\Pico{1}{2}};
    \node[anchor=north,yshift=\DOWN] at (axis cs:4,14) {\Pico{2}{2}};
    \node[anchor=north,yshift=\DOWN] at (axis cs:5,16) {\Pico{3}{2}};
    \node[anchor=north,yshift=\DOWN] at (axis cs:6,18) {\Pico{4}{2}};
    \node[anchor=north,yshift=\DOWN] at (axis cs:7,20) {\Pico{5}{2}};

    \node[anchor=north,yshift=\DOWN] at (axis cs:1,4)  {\Pico{1}{0}};
    \node[anchor=north,yshift=\DOWN] at (axis cs:2,8)  {\Pico{2}{0}};
    \node[anchor=south,yshift=\UP]   at (axis cs:3,12) {\Pico{3}{0}};
    \node[anchor=south,yshift=\UP]   at (axis cs:4,16) {\Pico{4}{0}};
    \node[anchor=south,yshift=\UP]   at (axis cs:5,20) {\Pico{5}{0}};
    \node[anchor=south,yshift=\UP]   at (axis cs:6,20) {\Pico{5}{1}};
    \node[anchor=south,yshift=\UP]   at (axis cs:7,20) {\Pico{5}{2}};

  \end{axis}
\end{tikzpicture}
\caption{Illustration of the lower bound construction in \Cref{sec:genLower}. The dashed orange curve corresponds to an incremental solution that favors the sets~$C_j$, while the solid green curve corresponds to an incremental solution that favors the sets~$R_i$. No incremental solution can simultaneously be better than~$1.25$-competitive at~$k=2$ and~$k=5$. Other solutions than the ones depicted are dominated by one of the curves.\label{fig:LB}}
\end{figure}
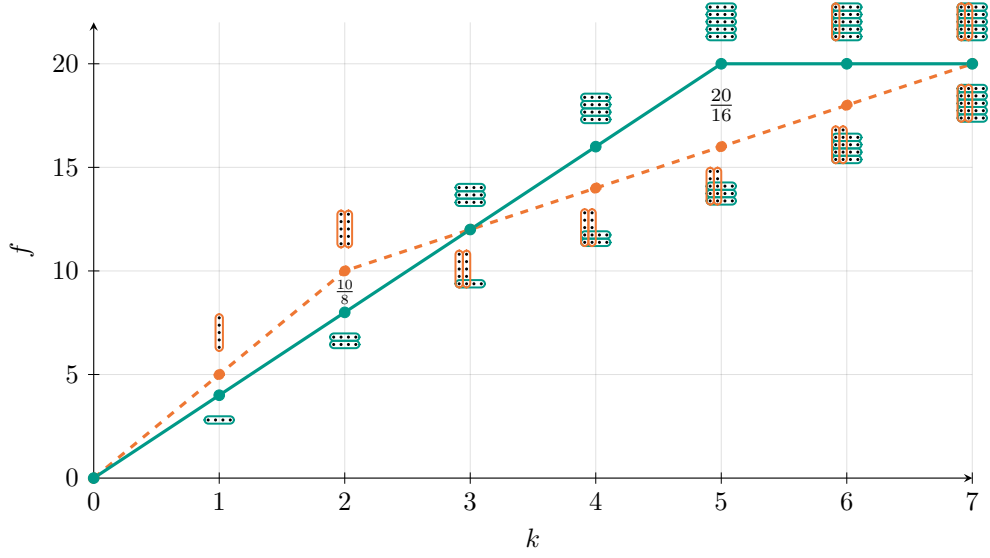

\begin{lemma}
    For every ordering $E_1,\dots,E_7$ of the subsets in $\U$, there exists $k \in [7]$ with
    \[ 
        \frac{\OPT{k}}{f(E_1 \cup \dots \cup E_k)} \geq 1.25. 
    \]
\end{lemma}

\begin{proof}
    We first observe that, by symmetry between the $C_j$'s and $R_i$'s, we may restrict our analysis to orderings that include $C_1,C_2$ and $R_1,R_2,R_3,R_4,R_5$ in these orders, respectively.
    Suppose that $\{E_1, E_2\} \neq \{C_1, C_2\}$, i.e.~that $C_1$ and $C_2$ are not the first two elements in the ordering $E_1,\dots,E_7$. 
    In this case,
    $f(E_1 \cup E_2) \leq \max\{ f(R_1, C_1), f(R_1, R_2) \} = 8,$
    and thus
    \[ \frac{\OPT{2}}{f(E_1 \cup E_2)} \geq \frac{f(C_1 \cup C_2)}{f(E_1 \cup E_2)} = \frac{10}{8} = \frac{5}{4}. \]
    
    Now, suppose that $\{E_1, E_2\} = \{C_1, C_2\}$. Then,
    \[ f(E_1 \cup \dots \cup E_5) = f(C_1 \cup C_2 \cup R_1 \cup R_2 \cup R_3) = 16, \]
    yielding
    \[ 
        \frac{\OPT{5}}{f(E_1 \cup \dots \cup E_5)} \geq \frac{f(R_1 \cup \cdots \cup R_5)}{f(E_1 \cup \dots \cup E_5)} = \frac{20}{16} = \frac{5}{4}.\qedhere
    \]
\end{proof}

\subsection{A Lower Bound Against Our Algorithm}\label{sec:algLower}

We begin by providing an informal overview of our lower bound in \Cref{sec:informal lower}, followed by a formal description of the lower bound instance in \Cref{sec:formal lower}.

\subsubsection{Overview of the Lower Bound for \Alg}\label{sec:informal lower}

To demonstrate that the analysis of our algorithm is tight, we construct a hard instance for it, based on a geometric coverage function $f$ defined over a collection of high-dimensional rectangles $\U$, so that $f(S)$ measures the volumes of the union of the rectangles in $S \subseteq \U$. The core idea is to exploit the density-based choice of \Alg by presenting it with a sequence of item groups that appear attractive locally but are inefficient globally.

\subparagraph{The Instance $(\U, f)$.}
More formally, the instance uses a ground set $\mathcal{U}$ of $(q+1)$-dimensional rectangles, organized into groups $\mathcal{U}^{(1)}, \dots, \mathcal{U}^{(q)}$. The geometry of these rectangles is defined such that elements within a single group $\mathcal{U}^{(i)}$ are disjoint; they effectively partition the $i^{th}$ dimension into $s_i$ segments and cover a volume of $v_i$. However, elements from different groups overlap significantly. This specific geometric arrangement ensures that while a single group covers a volume efficiently on its own, adding elements from a later group $\mathcal{U}^{(j)}$ to an existing solution consisting of elements from $\mathcal{U}^{(1)} \cup \dots \cup \mathcal{U}^{(j-1)}$ provides diminishing marginal gains due to the heavy spatial overlap with the volume already covered by previous groups. 

Intuitively, this construction generalizes the 2-dimensional hard instance from \Cref{sec:genLower} (visualized in \Cref{fig:LB}). In the 2-dimensional case, the instance offers `columns' ($C_j$) which are dense for small cardinalities, but ultimately inefficient compared to the `rows' ($R_i$) that form the optimal solution for larger cardinalities. Analogously, our high-dimensional instance provides groups of elements aligned with the $i^{th}$ dimension (group $\mathcal{U}^{(i)}$) which are locally dense, but less efficient than the groups aligned with the $(i+1)^{th}$ dimension (group $\mathcal{U}^{(i+1)}$) that constitute the optimal solution for the corresponding larger cardinality constraint $s_{i+1}$. This effectively generalizes the geometric conflict depicted in \Cref{fig:LB} to $q$ dimensions. In \Cref{sec:formal lower}, we provide the precise definition of this lower bound instance.

\subparagraph{The Behavior of \Alg on $(\U, f)$.}
By picking the appropriate values for $s_i$ and $v_i$, we can ensure that \Alg behaves very predictably. Due to the constructed densities, the algorithm selects elements sequentially, exhausting the first group $\mathcal{U}^{(1)}$ before moving to $\mathcal{U}^{(2)}$, and so on (see \Cref{lem:alglowerbound}). Consequently, the algorithm accumulates value slowly relative to the number of items chosen. It reaches an objective value of $v_i$ only after expending a cumulative cardinality of $\sum_{j=1}^{i} s_j$.

\subparagraph{The Structure of the Optimal Solution.}
In contrast, the optimum solution with value~$v_i$ is extremely efficient. 
For a cardinality constraint of~$s_i$, the optimum simply selects the group $\mathcal{U}^{(i)}$, achieving the objective value $v_i$ using only $s_i$ elements (see \Cref{cl:keypoints}). This discrepancy forces the algorithm to lag behind the optimum in terms of cardinality overhead. 
By analyzing the ratio between the values of the optimal solution and our algorithm's solution at the cardinality $k=s_q$, we establish a lower bound of $1.3724$ (see \Cref{lem:lower:1}).
\Cref{fig:hard_instance_lower_bound} provides an illustration of the lower bound obtained by comparing the optimal solution with the solution maintained by our algorithm on the instance $(\U, f)$.

\begin{figure}[t]
    \centering
    \includegraphics[width=0.96\textwidth]{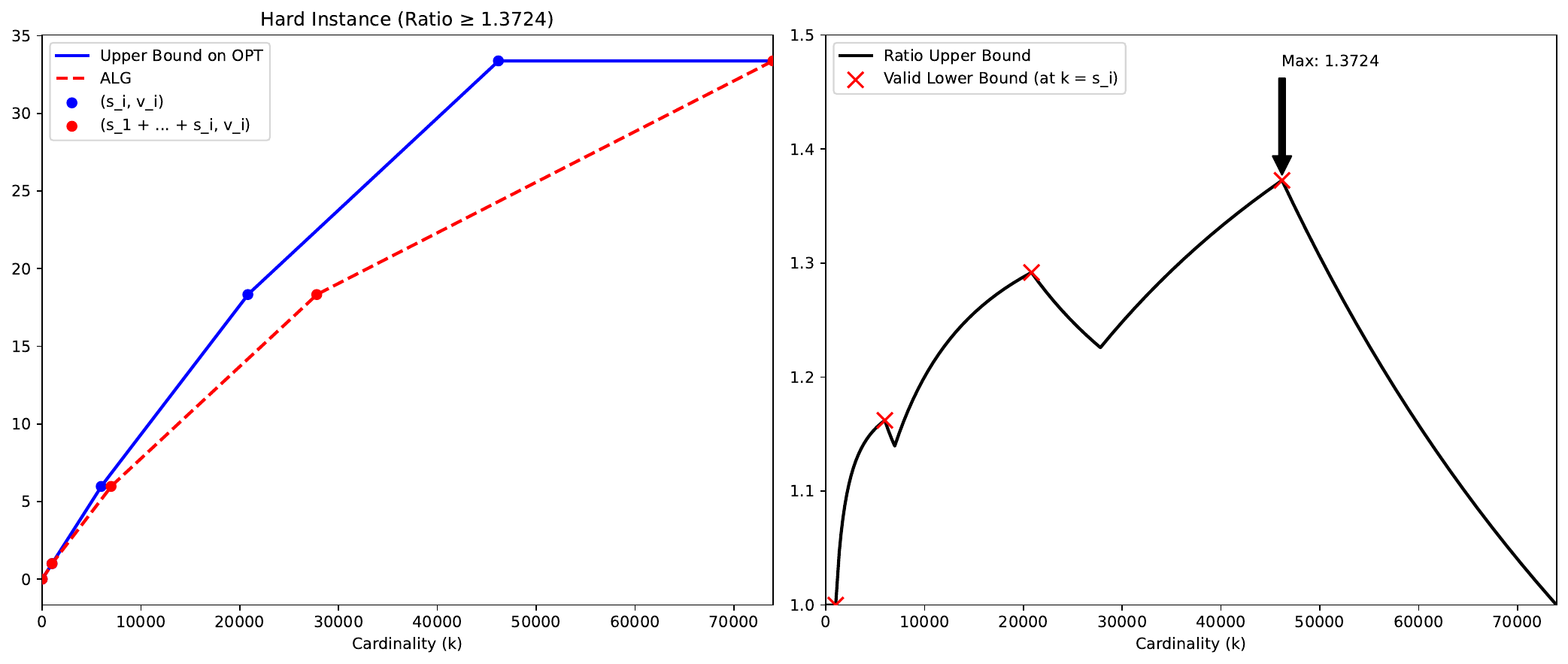}
    \caption{Visualization of the hard instance for our algorithm. 
    \textbf{Left:} The solid blue line depicts the concave envelope of the points $(s_i, v_i)$, which serves as an upper bound on $\OPT{k}$ for all $k$, and equals $\OPT{k}$ when $k = s_i$. 
    The dashed red line (\textsc{ALG}) shows the algorithm's performance in this instance; it reaches the value $v_i$ only after expending a cumulative cardinality of $\sum_{j=1}^i s_j$ (red dots). 
    \textbf{Right:} The ratio between the upper bound on $\OPT{k}$ and the algorithm's objective. The crosses mark the specific cardinalities $k=s_i$ where the upper bound coincides with~$\OPT{k}$, establishing the valid lower bound of $\geq 1.3724$.}
    \label{fig:hard_instance_lower_bound}
\end{figure}

\subsubsection{The Formal Lower Bound Instance}\label{sec:formal lower}

We begin by constructing a collection of subsets, which we will use in order to define a~submodular function. Let $v_1,\dots,v_q \in \R_{> 0}$ be a sequence of non-decreasing \emph{volumes} and $s_1,\dots,s_q \in \Z_{> 0}$ a sequence of non-decreasing \emph{quantities}. For each $i \in [q], \ell \in [s_i]$, we define the set $E^{(i)}_\ell \subseteq \R^{q + 1}$ as
\[ E^{(i)}_\ell \define [0, v_i) \times [0,1)^{i - 1} \times \left[\frac{\ell - 1}{s_i}, \frac{\ell}{s_i}\right) \times [0,1)^{q - i}. \]
We let $\U^{(i)} \define \{ E^{(i)}_\ell \, | \, \ell \in [s_i] \}$ and $\U \define \smash{\bigcup_{i \in [q]} \U^{(i)}}$ denote the collection of these $(q + 1)$-dimensional half-open rectangles. For notational convenience, we define $E^{(i)} \define \smash{\bigcup_{\ell \in [s_i]} E_\ell^{(i)}} = [0,v_i) \times [0,1)^q$, $v_0 \define 0$, $s_0 \define 0$, and let $n \define |{\U}| = \sum_{i=1}^{q} s_i$ denote the size of the collection~$\U$.

\begin{restatable}{lemma}{lemhighdimboxes}
\label{lem:highdimboxes}
The collection of subsets ${\U}$ satisfies the following properties:
\begin{enumerate}
    \item For all $1 \leq i \leq q$, $\ell, \ell' \in [s_i]$, we have $\vol \!\big(E_\ell^{(i)} \big) = v_i / s_i$ and $E_\ell^{(i)} \cap E_{\ell'}^{(i)} = \emptyset$ if $\ell \neq \ell'$.\label{HDlem:P1}
    \item For all $1 < i \leq q$, $C$ an arbitrary union of subsets in $\U^{(1)} \cup \dots \cup \U^{(i-1)}$, and $E \in \U^{(i)}$, we have that $\vol (E \cap C) = \vol(C) / s_i$.
\end{enumerate}
\end{restatable}

\begin{proof}    
	For each $i \in [q]$, $\ell \in [s_i]$, we can see that
	\[\vol \! \big(E_\ell^{(i)} \big) = v_i \cdot 1^{i - 1} \cdot \frac{1}{s_i} \cdot 1^{q - i} = \frac{v_i}{s_i}.\]
	Additionally, for any $\ell' \in [s_i]$ such that $\ell \neq \ell'$, we have that
	\[E^{(i)}_\ell \cap E^{(i)}_{\ell'} = [0, v_i) \times [0,1)^{i - 1} \times \emptyset \times [0,1)^{q - i} = \emptyset. \]
	
	\noindent
	Now, consider some $1 < i \leq q$, a subset $C$ obtained by taking an arbitrary union of subsets in $\U^{(1)} \cup \dots \cup \U^{(i-1)}$, and $E = E_\ell^{(i)} \in \U^{(i)}$. Let $C' \subseteq [0, v_{i-1}) \times [0,1)^{i-1}$ denote the subset such that $C = C' \times [0,1)^{q - i + 1}$. We can see that
	\[ E \cap C = E \cap (C' \times [0,1)^{q - i + 1}) = C' \times \left[\frac{\ell - 1}{s_i}, \frac{\ell}{s_i}\right) \times [0,1)^{q - i}. \]
	Thus, we have that
	\[ 
	\vol(E \cap C) = \vol(C') \cdot \frac{1}{s_i} = \vol(C) \cdot \frac{1}{s_i}.\qedhere
	\]
\end{proof}

\subparagraph{The Lower Bound Instance.}
Consider sequences of non-decreasing volumes $v_1,\dots,v_q \in \R_{> 0}$ and non-decreasing quantities $s_1,\dots,s_q \in \Z_{> 0}$, satisfying the following condition for each $1 \leq i < q$:
\begin{equation}\label{eq:lower bnd prop}
    \frac{v_{i+1} - v_{i-1}}{s_{i+1}} \leq \frac{v_i - v_{i-1}}{s_i}.
\end{equation}
Let $\U$ be the collection of subsets defined by these sequences,
and let $f$ be the geometric coverage function defined over the collection of subsets $\U$. In other words, for some $S \subseteq \U$, we have $f(S) = \vol \! \left(\bigcup_{E \in S} E \right)$.

\subparagraph{Our Lower Bound.}

The following lemma, which we prove in \Cref{sec:finallowerproof}, formalizes the intuition described in \Cref{sec:informal lower}.

\begin{lemma}\label{lem:lower:1}
    For the objective~$f$ defined above such that \Cref{eq:lower bnd prop} holds, \Cref{alg:main} produces a solution with competitive ratio at least
    \[ v_q \cdot \left(v_{i^\star} + \frac{r^\star}{s_{i^\star + 1}} \cdot (v_{i^\star + 1} - v_{i^\star}) \right)^{-1}, \]
    where $i^\star \in [q]$ is the unique value such that $s_1 + \dots + s_{i^\star} + r^\star = s_q$, for $r^\star \leq s_{i^\star + 1}$.
\end{lemma}

By plugging in some specific values\footnote{Which we found using a computer.} for the $v_i$ and $s_i$ in \Cref{lem:lower:1} shown in \Cref{sec:finallowerproof}, we can prove \Cref{app:thm:L2}.

\begin{proof}[Proof of \Cref{app:thm:L2}]
	Consider the following sequences for $q = 4$:
	\[ (s_i)_{i = 1}^q = (10000, 59683, 208018, 461378), \]
	\[ (v_i)_{i = 1}^q = (1.0000, 5.9683, 18.3164, 33.3561). \]
	It is easy to verify that these sequences satisfy the condition from \Cref{eq:lower bnd prop}. We can also compute the values of $i^\star = 3$ and $r^\star = 183677$ from \Cref{lem:lower:1}, yielding a lower bound on the competitive ratio of at least $1.3724$.
	\Cref{fig:hard_instance_lower_bound} gives a plot with these values.
\end{proof}

\subsection{Proof of \Cref{lem:lower:1}}
\label{sec:finallowerproof}

We begin by deriving some useful properties of the instance, then proceed to analyze the structure of the optimal solution, followed by an analysis of the solution constructed by \Alg.

\subsubsection{Useful Properties of the Instance}

The following claims are immediate applications of \Cref{eq:lower bnd prop}.

\begin{restatable}{claim}{claimextrastructure}
	\label{cl:extra structre}
	For each $1 \leq i \leq j \leq q$, we have that
	\[\frac{v_{j} - v_{i-1}}{s_{j}} \leq \frac{v_i - v_{i-1}}{s_i}.\]
\end{restatable}
\begin{proof}
	Let $i \in [q]$ and, for each $i \leq j \leq q$, define
	\[ R_j \define \frac{v_{j} - v_{i-1}}{s_j}. \]
	We prove by induction that $R_i \geq R_{i + 1} \geq \dots \geq R_q$. Suppose we have that $R_i \geq \dots \geq R_j$, for some $i \leq j < q$. Then we have that
	\begin{align*} R_{j + 1} &= \frac{v_{j + 1} - v_{i-1}}{s_{j + 1}} = \frac{v_{j + 1} - v_{j-1}}{s_{j + 1}} + \frac{v_{j-1} - v_{i-1}}{s_{j + 1}}\\ 
		&\leq \frac{v_{j} - v_{j-1}}{s_{j}} + \frac{v_{j-1} - v_{i-1}}{s_{j}} = \frac{v_{j} - v_{i-1}}{s_{j}} = R_j, 
	\end{align*}
	where we are applying \Cref{eq:lower bnd prop} and the fact that $s_j \leq s_{j+1}$.
\end{proof}

\begin{restatable}{claim}{claimextrastructuretwo}
	\label{cl:extra structre 2}
	For each $1 \leq i < q$, we have that $v_{i}/s_{i} \geq v_{i+1}/s_{i+1}$.
\end{restatable}
\begin{proof}
	We have that
	\[\frac{v_i}{s_i} - \frac{v_{i+1}}{s_{i+1}} = \left( \frac{v_i - v_{i-1}}{s_i} - \frac{v_{i+1} - v_{i-1}}{s_{i+1}} \right) + v_{i-1} \cdot\left( \frac{1}{s_i} - \frac{1}{s_{i+1}} \right) \geq 0,\]
	where the first term is greater than $0$ by \Cref{eq:lower bnd prop}, and the second term since $s_{i + 1} \geq s_i$.
\end{proof}

\subsubsection{Analysis of the Optimal Solution}\label{sec:lower:opt:anal}

Given a subset $S \subseteq \U$, we define $\chi(S) \define \max \big\{j \in [q] \, \mid \, S \cap \U^{(j)} \neq \emptyset\big\}$, i.e., $\chi(S)$ is the maximum index $j \in [q]$ such that $S$ contains a set from $\U^{(j)}$. 
Additionally, for each $j \in [q]$, we define
$S_{\leq j} \define S \cap ( {\U}^{(0)} \cup \dots \cup {\U}^{(j)} )$.

We begin by deriving the following useful identity for each $S \subseteq \U$.

\begin{restatable}{claim}{claimeqoptanalysis} 
	For each $j \in [q]$, we have that
	\begin{equation}\label{eq:optanalysis}
		f(S_{\leq j}) = \frac{|S \cap {\U}^{(j)}|}{s_{j}} \cdot \left(v_{j} - f(S_{\leq j - 1}) \right) + f(S_{\leq j - 1}).
	\end{equation}
\end{restatable}
\begin{proof}
	For each $j' \in [q]$, let $a_{j'} \define |S \cap {\U}^{(j')}|$.
	Let $C \define \bigcup_{E \in S_{\leq j -1}} E$ and $B = \bigcup_{E \in S \cap \U^{(j)}} E$ denote the union of the sets in $S_{\leq j - 1}$ and $S_{\leq j} \setminus S_{\leq j - 1}$, respectively. 
	Applying \Cref{lem:highdimboxes} and $\vol(C) = f(S_{\leq j -1})$ by definition, we have that
	\begin{align*} f(S_{\leq j}) = \vol(B \cup C) &= \vol(B) + \vol(C) - \vol(B \cap C) \\
		&= \frac{a_{j}}{s_{j}} \cdot \vol(E^{(j)}) + \vol(C) - \sum_{E \in S\cap\U^{(j)}} \frac{1}{s_{j}} \cdot \vol(E \cap C) \\
		&= \frac{a_{j}}{s_{j}} \cdot \vol(E^{(j)}) + f(S_{\leq j -1}) - \frac{a_{j}}{s_{j}} \cdot \vol(C) \\
		&= \frac{a_{j}}{s_{j}} \cdot \left(v_{j} - f(S_{\leq j - 1}) \right) + f(S_{\leq j - 1}). \qedhere 
	\end{align*}
\end{proof}

Using \Cref{eq:optanalysis}, we now prove the following key lemma.

\begin{restatable}{lemma}{claimpushdown}
	\label{cl:pushdown}
	Let $S \subseteq \U$ be a subset of size at most $s_i$. If $\chi(S) > i$, then there exists a~subset $S' \subseteq \U$ of size $|S|$ such that $f(S') \geq f(S)$ and $\chi(S') < \chi(S)$.
\end{restatable}
\begin{proof}        
	For each $j' \in [q]$, let $a_{j'} \define |S \cap {\U}^{(j')}|$.
	Now, let $j \define \chi(S)$. Let $S'$ be the set obtained by adding any $a_j$ elements from $\U^{(j-1)} \setminus S_{\leq j -1}$ to $S_{\leq j -1}$. Note that $|\U^{(j-1)} \setminus S_{\leq j -1}| \geq a_j$, so this is always possible. Clearly, $\chi(S') = j - 1 < j = \chi(S)$. It remains to show that $f(S') \geq f(S)$.
	
	\subparagraph{Proof that $f(S') \geq f(S)$.}
	Similarly to \Cref{eq:optanalysis}, by applying \Cref{lem:highdimboxes}, we can obtain the following identity:
	\begin{equation}
		f(S') = \frac{a_{j}}{s_{j - 1}} \cdot \left(v_{j - 1} - f(S_{\leq j - 2}) \right) + f(S_{\leq j - 1}).
	\end{equation}
	Now we can lower bound $f(S') - f(S)$ as follows:
	\begin{align*}
		f(S') - f(S) &= f(S') - f(S_{\leq j})\\
		&=\frac{a_{j}}{s_{j - 1}} \cdot \left(v_{j - 1} - f(S_{\leq j - 2}) \right) + f(S_{\leq j - 1}) - \frac{a_{j}}{s_{j}} \cdot \left(v_{j} - f(S_{\leq j - 1}) \right) - f(S_{\leq j - 1})\\
		&=\frac{a_{j}}{s_{j - 1}} \cdot \left(v_{j - 1} - f(S_{\leq j - 2}) \right) - \frac{a_{j}}{s_{j}} \cdot \left(v_{j} - f(S_{\leq j - 1}) \right)\\
		&\geq\frac{a_{j}}{s_{j - 1}} \cdot \left(v_{j - 1} - f(S_{\leq j - 2}) \right) - \frac{a_{j}}{s_{j}} \cdot \left(v_{j} - f(S_{\leq j - 2}) \right)\\
		&= a_j \cdot \left( \frac{v_{j-1}}{s_{j-1}} - \frac{v_{j}}{s_{j}} + \left(\frac{1}{s_j} - \frac{1}{s_{j-1}} \right) \cdot f(S_{\leq j - 2}) \right)\\
		&\geq a_j \cdot \left( \frac{v_{j-1}}{s_{j-1}} - \frac{v_{j}}{s_{j}} + \left(\frac{1}{s_j} - \frac{1}{s_{j-1}} \right) \cdot v_{j-2} \right)\\
		&= a_j \cdot \left( \frac{v_{j-1} - v_{j-2}}{s_{j-1}} - \frac{v_{j} - v_{j-2}}{s_{j}} \right) \geq 0.
	\end{align*}
	The first inequality follows from $f(S_{\leq j -2}) \leq f(S_{\leq j - 1})$ by monotonicity. The second inequality follows from $f(S_{\leq j -2}) \leq f(\U^{(j-2)}) = v_{j-2}$ and the observation that $(1/s_j - 1/s_{j-1}) \leq 0$ since $s_j \geq s_{j-1}$. The third inequality follows from \Cref{eq:lower bnd prop}. This concludes the proof.
\end{proof}

The following lemma shows that $\U^{(i)}$ is an optimal solution for $k = s_i$.

\begin{restatable}{lemma}{lemmakeypoints}
	\label{cl:keypoints}
	For each $i \in [q]$, we have that $\OPT{s_i} = v_i$.
\end{restatable}
\begin{proof}
	We have that $v_i = \vol(E^{(i)}) = f(\U^{(i)}) \leq \OPT{s_i}$, since~$|\U^{(i)}| = s_i$. Thus, it suffices to show that $\OPT{s_i} \leq v_i$. 
	
	Given any subset $S \subseteq \U$ of size $s_i$, we can repeatedly apply \Cref{cl:pushdown} to the set $S$ to obtain a subset $S' \subseteq \U$ of size $s_i$ such that $f(S') \geq f(S)$ and $\chi(S') \leq i$. 
	Since we have that $\bigcup_{E \in S'} E \subseteq \bigcup_{j=1}^i E^{(j)}= E^{(i)}$,
	it follows that
	\[f(S) \leq f(S') \leq f({\U}^{(i)}) = v_i.\qedhere\]
\end{proof}

\begin{restatable}{lemma}{lemmaupperOPT}
	\label{cl:upperOPT}
	For each $k \in [n]$ such that $s_{i-1} \leq k \leq s_{i}$, we have that
	\[ \OPT{k} \leq v_{i-1} + \frac{k - s_{i-1}}{s_{i} - s_{i-1}} \cdot (v_{i} - v_{i-1}). \]
\end{restatable}
\begin{proof}
	Let $k \in [n]$ such that $s_{i-1} \leq k \leq s_{i}$, and define
	\[ L(k) \define v_{i-1} + \frac{k - s_{i-1}}{s_{i} - s_{i-1}} \cdot (v_{i} - v_{i-1}). \]
	Let $S^\star \subseteq \U$ be a subset of size $k$ such that $f(S^\star) = \OPT{k}$. By repeatedly applying \Cref{cl:pushdown}, we can assume w.l.o.g.~that $\chi(S^\star) \leq i$. Applying \Cref{eq:optanalysis} and letting $a_i \define |S \cap \U^{(i)}|$, we have that:
	\[ f(S^\star) = \frac{a_i}{s_i} \cdot v_i + \left(1 - \frac{a_i}{s_i} \right) \cdot f(S^\star_{\leq i - 1}),\]
	Now, let $x \define s_{i-1} - |S^\star_{\leq i-1}|$. Since $S^\star$ is an optimal solution of size $k$, we must have $x \geq 0$, otherwise we could find a solution of size $k$ with a greater objective by replacing the elements in $S^\star_{\leq i-1}$ with $\U^{(i-1)}$ and adding $-x$ new elements to the solution from, say, $\U^{(q)}$. Since $k = (s_{i-1} - x) + a_i$, it follows that
	\[f(S^\star) = \frac{(k - s_{i-1}) + x}{s_i} \cdot v_i + \left(1 - \frac{(k - s_{i-1}) + x}{s_i} \right) \cdot f(S^\star_{\leq i - 1}). \]
	Finally, noting that at least $x$ element from $\U^{(i-1)}$ are not contained in $S^\star_{\leq i-1}$, we get that
	\[ f(S^\star_{\leq i-1}) \leq v_{i-1} - \frac{x}{s_{i-1}} \cdot (v_{i-1} - v_{i-2}). \]
	Combining these equations, we get the following upper bound on $f(S^\star)$:
	\[ f(S^\star) \leq Q \define \frac{(k - s_{i-1}) + x}{s_i} \cdot v_i + \left(1 - \frac{(k - s_{i-1}) + x}{s_i} \right) \cdot \left( v_{i-1} - \frac{x}{s_{i-1}} \cdot (v_{i-1} - v_{i-2}) \right). \]
	Treating this upper bound $Q$ as a function of $x$ and fixing all other variables, we can see that $Q(x)$ is a quadratic function in $x$, where the coefficient of $x^2$ is positive. It follows that $Q(x)$ is a convex function on the closed interval $[0, s_{i-1}]$, and thus achieves its maximum value at one of the endpoints of the interval:
	\[ \max_{0 \leq x \leq s_{i-1}} Q(x) = \max \{Q(0), Q(s_{i-1}) \}. \]
	We now show that $Q(0) \leq L(k)$ and $Q(s_{i-1}) \leq L(k)$, which concludes the proof.
	
	\subparagraph{Proof that $Q(0) \leq L(k)$.}
	\[ Q(0) = \frac{k - s_{i-1}}{s_i} \cdot v_i + \left(1 - \frac{k - s_{i-1}}{s_i} \right) \cdot v_{i-1} = v_{i-1} + \frac{k - s_{i-1}}{s_{i}} \cdot (v_{i} - v_{i-1}) \leq L(k), \]
	where the inequality follows from $s_i - s_{i-1} \leq s_i$.
	
	\subparagraph{Proof that $Q(s_{i-1}) \leq L(k)$.} Consider the line $L'$, as a function of $k$, defined by $Q(s_{i-1})$: 
	\[ L'(k) \define Q(s_{i-1}) = \frac{k}{s_i} \cdot v_i + \left(1 - \frac{k}{s_i} \right) \cdot v_{i-2}. \]
	To show that $Q(s_{i-1}) \leq L(k)$, it suffices to show that $L'(k) \leq L(k)$ for all $s_{i-1} \leq k \leq s_i$, which is true if $L'(s_{i-1}) \leq L(s_{i-1})$ and $L'(s_{i}) \leq L(s_{i})$, i.e.~if $L$ is greater than $L'$ at the endpoints of the interval $[s_{i-1}, s_i]$. Clearly, $L'(s_i) = v_i = L(s_i)$, and
	\[L'(s_{i-1}) = \frac{s_{i-1}}{s_i} \cdot v_i + \left(1 - \frac{s_{i-1}}{s_i} \right) \cdot v_{i-2} = v_{i-2} + s_{i-1} \cdot \left(\frac{v_i - v_{i-2}}{s_i}\right) \]\[ \leq v_{i-2} + s_{i-1} \cdot \left(\frac{v_{i-1} - v_{i-2}}{s_{i-1}}\right)  = v_{i-1} = L(s_{i-1}),  \]
	where the inequality follows from \Cref{eq:lower bnd prop}.  
\end{proof}

\begin{restatable}{claim}{claimbadchoice}
	\label{cl:bad choice}
	For each $i \in [q]$, we have that
	\[ \max_{k \in [n]} \left( \frac{\OPT{k} - v_{i-1}}{k} \right) = \frac{v_i - v_{i-1}}{s_i}. \]
\end{restatable}
\begin{proof}
	Let $R(k) \define (\OPT{k} - v_{i-1})/k$ denote the ratio we wish to maximize.
	We begin by showing that $R(k)$ attains its maximum value for some $k$ in $\{s_1, \dots, s_q\} \subseteq [n]$.
	
	\subparagraph{Reducing the Search Space.}
	We first note that $R(k)$ attains its maximum for some $k \leq s_q$ since $\OPT{s_q} \geq f(E^{(q)}) = \OPT{n}$. Furthermore, for $k \leq s_1$, we have that
	\[ R(k) = \frac{k \cdot (v_1 / s_1) - v_{i-1}}{k} = \frac{v_1}{s_1} - \frac{v_{i-1}}{k} \leq R(s_1), \]
	so $R(k)$ attains its maximum for some $k \geq s_1$.
	Now, consider some $k \in [s_1, s_q]$, and let $j > 1$ be an index such that $s_{j-1} \leq k \leq s_j$. By \Cref{cl:upperOPT}, we have that
	\[\OPT{k} \le L(k) \define v_{j-1} + \frac{k - s_{j-1}}{s_{j} - s_{j-1}}(v_{j} - v_{j-1}).\]
	Now, consider the following upper bound on $R(k)$:
	\[\tilde{R}(k) \define \frac{L(k) - v_{i-1}}{k}.\]
	Within the interval $[s_{j-1}, s_{j}]$, $L(k)$ is a line of the form $Ak + B$, so it follows that $\tilde{R}(k)$ can be expressed as $A + C/k$, which is a monotonic function of $k$. Since a monotonic function defined on a closed interval attains its maximum value at one of the endpoints, it follows that, for any $k$ in this interval,
	\[R(k) \le \tilde{R}(k) \le \max \left( \tilde{R}(s_{j-1}), \tilde{R}(s_{j}) \right)\]
	Since $\tilde{R}(s_{j-1}) = R(s_{j-1})$ and $\tilde{R}(s_{j}) = R(s_{j})$ by \Cref{cl:keypoints}, we have that
	\[\max_{k \in [s_{j - 1}, s_j]} R(k) = \max \{ R(s_{j-1}), R(s_{j}) \}.\]
	Consequently, the global maximum of $R(k)$ must occur for some $k \in \{s_1, \dots, s_q\}$.
	
	\subparagraph{Identifying the Maximum Value.} Now, consider some $j < i$. Since the $v_j$ are non-decreasing, we have that
	\[R(s_j) = (v_{j} - v_{i-1})/s_j \leq 0 \leq (v_{i} - v_{i-1})/s_i = R(s_i).\]
	For $j > i$, it follows from \Cref{cl:extra structre} that $R(s_j) \leq R(s_i)$. Thus, we have that
	\[ \max_{1 \leq j \leq q} R(s_j) = R(s_i). \qedhere \]
\end{proof}

\subsubsection{Analysis of \Alg}\label{sec:lower:alg:anal}

The following lemma summarizes the behavior of \Alg on this lower bound instance.

\begin{restatable}{lemma}{lemmaalglowerbound}
	\label{lem:alglowerbound}
	Given the input $(\U, f)$, \Alg runs for $q$ phases, where the solution maintained by the algorithm at the end of the $i^{th}$ phase is $\U^{(1)} \cup \dots \cup \U^{(i)} \subseteq \U$.
\end{restatable}
\begin{proof}
	We prove this lemma by induction on the phase $i$.
	
	\subparagraph{Inductive Step.}
	Let $i \in [q]$ and consider the state of the algorithm at the start of phase~$i$. Assume that the length of each phase $j < i$, $k_j$, was exactly $s_j$, and that the subset of elements added to the solution $S$ during this phase was $\U^{(j)}$. Thus, the solution $S$ at the start of phase $i$ is $\U^{(1)} \cup \dots \cup \U^{(j)}$.\footnote{If $i = 0$, the solution $S$ is empty, and we define $k_0 = 0$.} Now, the algorithm chooses $k$ that maximizes the value of
	\begin{align*}
		\frac{\OPT{k} - \OPT{k_{i-1}}}{k} &= \frac{\OPT{k} - v_{i-1}}{k} 
		&& \textnormal{(by \Cref{cl:keypoints})}
	\end{align*}
	and sets this to be $k_i$.
	By \Cref{cl:bad choice}, the maximum value of this expression is $(v_i - v_{i-1})/s_i$, which is precisely the value of this expression when $k = s_i$. Thus, the algorithm sets $k_i=s_i$.\footnote{We assume that we break ties arbitrarily in our favor.} Since $\OPT{s_i} = v_i = f(\U^{(i)})$, we have that $S^\star_{k_i} = \U^{(i)}$. The phase continues until $S^\star_{k_i} \subseteq S$. Since $\U^{(i)}$ is disjoint from the solution $S$ at the start of the phase, the phase goes on for $s_i$ iterations, and terminates with $S = \U^{(1)} \cup \dots \cup \U^{(i)}$.
	
	\subparagraph{Conclusion.} At the end of phase $q$, we have that $S = \U^{(1)} \cup \dots \cup \U^{(q)} = \U$, and the algorithm terminates.
\end{proof}

\subsubsection{Concluding the Proof of \Cref{lem:lower:1}} 

We are now ready to prove \Cref{lem:lower:1} by upper bounding the objective of our algorithm and lower bounding the value of the optimal solution at some point during the execution of the algorithm.

\subparagraph{Upper Bounding Our Objective.}
Let $i^\star \in [q]$ be the unique value such that $s_1 + \dots + s_{i^\star} + r^\star = s_q$, for $r^\star \leq s_{i^\star + 1}$, and let $\tau = s_q$.
Consider the subset $S$ consisting of the first~$\tau$ elements in the incremental solution returned by \Alg when given $(f, \U)$ as input.
It follows from \Cref{lem:alglowerbound} that
\[ {\U}^{(1)} \cup \dots \cup {\U}^{(i^\star)} \subseteq S \subseteq {\U}^{(1)} \cup \dots \cup {\U}^{(i^\star + 1)}.\]
Thus, by \Cref{lem:highdimboxes}, we have that
\[ f(S) = \vol \! \left( E^{(i^\star)} \right) + \frac{r^\star}{s_{i^\star + 1}} \cdot \left(\vol \! \left( E^{(i^\star + 1)} \right) - \vol \! \left( E^{(i^\star)} \right)\right) = v_{i^\star} + \frac{r^\star}{s_{i^\star + 1}} \cdot \left(v_{i^\star + 1} - v_{i^\star}\right). \]

\subparagraph{Lower Bounding the Optimum.}
It follows from \Cref{cl:keypoints} that $\OPT{\tau} = \OPT{s_q} = v_q$.

\subparagraph{Lower Bounding the Competitive Ratio.} 
Putting these bounds together, it follows that the competitive ratio of our algorithm is at least
\[ \frac{\OPT{\tau}}{f(S)} = v_q \cdot \left(v_{i^\star} + \frac{r^\star}{s_{i^\star + 1}} \cdot (v_{i^\star + 1} - v_{i^\star}) \right)^{-1}.  \]

%% file: appendix.tex
\section{\texorpdfstring{Missing Proofs in Upper Bound Analysis (\Cref{sec:analysis})}{Missing Proofs in Upper Bound Analysis}}
\label{sec:missing_proofs}

\submodularitygivesatleastaverage*

\begin{proof}
    Let $m \define |B| - |A|$. Fix an arbitrary increasing sequence of $m+1$ sets $A = A_0 \subsetneq A_1 \subsetneq \cdots \subsetneq A_m = B$. Clearly, $f(B) - f(A) = \sum_{i=1}^{m} f(A_i) - f(A_{i-1})$.

    By the averaging argument, there exists $q \in [m]$ such that $f(A_q) - f(A_{q-1}) \geq (f(B) - f(A)) / m$. Let $e$ be the sole element of $A_q \setminus A_{q-1}$.
    Then, 
    \[
        f(A \cup \set{e}) - f(A) 
            \geq f(A_{q-1} \cup \set{e}) - f(A_{q-1})
            = f(A_q) - f(A_{q-1})
            \geq \frac{f(B) - f(A)}{m},
    \]
    where the first inequality follows by the submodularity of $f$.
\end{proof}


\subsection{\texorpdfstring{Properties of Function $h$}{Properties of Function h}}

\label{app:conditions_on_h}
In this section we check that the function 
\[
    h(q) \define \frac{156+1543 \cdot q-550 \cdot q^2-149 \cdot q^3}{1000 \cdot q}.
\]
fulfills all properties of \Cref{cla:properties_of_h}. We restate them here for convenience.

\propertiesofh*

\begin{proof}[Proof of \Cref{prop:h_decreasing} of \Cref{cla:properties_of_h}]
    The property $h(1) = 1$ can be easily verified. 
    The first derivative of $h$ is given by
    \begin{align*}
        h'(q)
        & = \frac{(1543-1100q-447q^2) \cdot q-(156+1543q-550q^2-149q^3)}{1000 \cdot q^2} \\
        & = -\frac{156+550q^2+298q^3}{1000 \cdot q^2}.
    \end{align*} 
    As $h'(q) < 0$ for $q\geq0$, the function $h(q)$ is decreasing for $q \geq 0$.
\end{proof}

\begin{figure}
    \centering
    \includegraphics[width=0.5\linewidth]{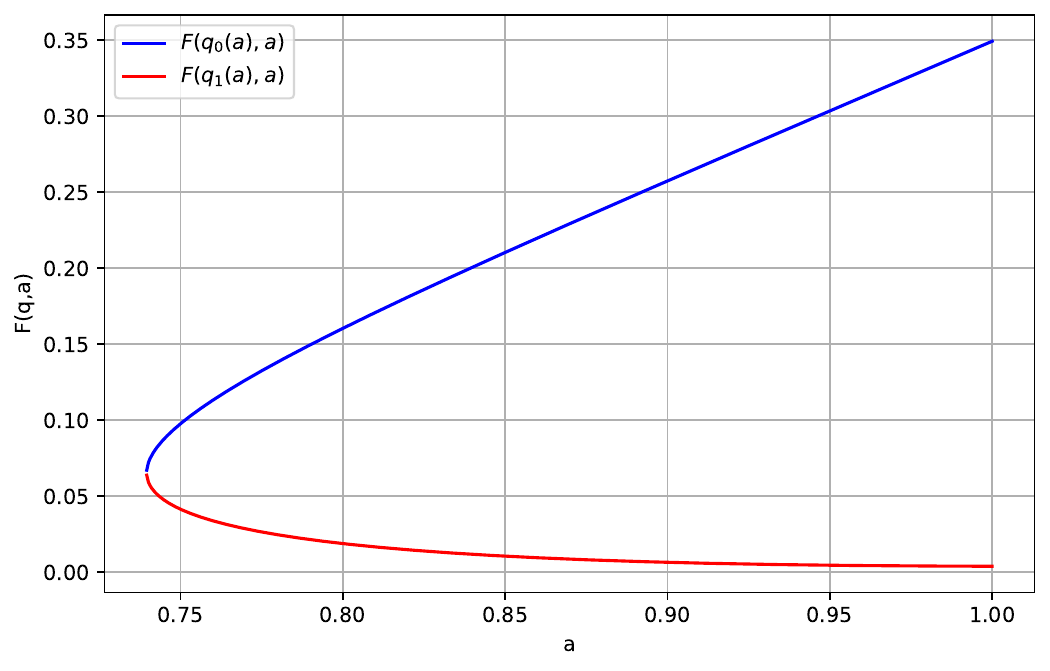}
    \caption{The value of $F(a,q)\define q(h(a)-1)+\frac{a}{1-a}(h(a)-1)-q h(q)$ at the local extreme points $q=q_0(a)$ and $q=q_1(a)$ in $(0,1)$ for fixed $a\in(a^*,1)$. We show $F(a,q) \geq 0$ for all $a\in (0,1)$ and all $q \in (0,1]$ to verify \Cref{prop:first_condition_on_h} of \Cref{cla:properties_of_h}.}
    \label{fig:mathematica_proof}
\end{figure}

\begin{proof}[Proof of \Cref{prop:first_condition_on_h} of \Cref{cla:properties_of_h}]
    Let $a=\frac{xq}{1+xq}$. Then, $x=\frac{a}{q(1-a)}$ for $a \in (0,1)$. Therefore, it suffices to show for $q \in (0,1]$ and $a \in (0,1)$ that
    \begin{align*}
        1+ \frac{qh(q)}{q+\frac{a}{1-a}} \leq h(a),
    \end{align*}
    which we rearrange to $F(a,q)\define q(h(a)-1)+\frac{a}{1-a}(h(a)-1)-q h(q) \geq 0$.
    
    For $q=1$, we have $(1-a)F(a,1)=(h(a)-1)-(1-a)=h(a)-2+a$. Further, we consider $a(1-a)F(a,1)=ah(a)-2a+a^2=(156-457a+450a^2-149a^3)/1000$ which has only one real root at $a_0=1$ and evaluates positive for $a=0$ which implies $F(a,1)>0$ for all~$a\in(0,1)$.
    Therefore, the property holds for $q=1$.

    For the other boundary case, where $q$ approaches $0$, we observe that $(1-a)F(a,q)$ approaches $\lim_{q \rightarrow 0} (1-a)F(a,q) = a(h(a)-1)-(1-a)156/1000=(699a-550a^2-149a^3)/1000$. This polynomial has one negative root and roots at $a_0=0$ and $a_1=1$. Since it evaluates positive for $a$ approaching $0$, we derive $\lim_{q \rightarrow 0} F(a,q)>0$ for all $a\in(0,1)$ and, thus, the property holds for $q>0$ small enough.
    
    For fixed $a \in (0,1)$, the first derivative of $F_a(q)\define q(h(a)-1)+\frac{a}{1-a}(h(a)-1)-q h(q)$ is
    \begin{align*}
        F'_a (q)=h(a)-1-h(q)-qh'(q).
    \end{align*}
    The quartic function $q^2F'_a(q)$ has no roots in $(0,1)$ as long as $a\leq a^* \approx 0.7395$. 
    
    For $a> a^*$, there are two real roots $q_0(a)$ and $q_1(a)$ in $(0,1]$ with $q_0(a)<q_1(a)$. One can verify at these local extreme points that $F(a,q_0(a))>0$ and $F(a,q_1(a))>0$ for all $a>a^*$, see \Cref{fig:mathematica_proof}.
    Thus, the property holds for all $q \in (0,1]$.
    
    Alternatively, the following Mathematica code confirms our geometric analysis:
    \begin{verbatim}
        h[z_] \define (156 + 1543 z - 550 z^2 -149 z^3)/(1000 z);
        Resolve[
            ForAll{q,x},
                0 < q <= 1 && x > 0,
                1 + h[q]/(1 + x) <= h[(q x)/(1 + q x)]
            ],
            Reals
        ] 
    \end{verbatim}
\end{proof}

\begin{proof}[Proof of \Cref{prop:second_condition_on_h} of \Cref{cla:properties_of_h}]
    Let $H(q) \define 1+ (h(q)-1) \cdot q = 1.156+0.543 \cdot q-0.55 \cdot q^2-0.149 \cdot q^3$. We want to show that $H(q) \leq \rho$ for all $q \in (0,1]$.
    
    The first derivative of $H$ is $H'(q)=0.543-1.1 \cdot q-0.447 \cdot q^2$ and the second derivative is $H''(q)=-1.1-0.894 \cdot q$.
    $H'(q) = 0$ has two roots, $q_0$ and $q_1$, given by
    \begin{align*}
        q_0,q_1 & = - \frac{550}{447} \pm \sqrt{\left(\frac{550}{447}\right)^2+\frac{543}{447}}.
    \end{align*}
    We have $q_0 \approx 0.4215 \in (0,1]$ and $q_1 < 0$.
    Since $H''(q)<0$ for all $q \in (0,1]$, $H(q)$ attains its maximum over interval $(0,1]$ at $q_0$ and thus $1 + (h(q)-1) \cdot q \leq H(q_0) < 1.2765 < \rho$. 
\end{proof}

\begin{proof}[Proof of \Cref{prop:third_condition_on_h} of \Cref{cla:properties_of_h}]
    We first analyze a helper function $F: [0,1] \to \mathbb{R}$ defined as
    \[ 
        F(s) \define \frac{h(s^2) \cdot s^2}{(1+s)^2}
        = \frac{156 + 1543 \cdot s^2 - 550 \cdot s^4 - 149 \cdot s^6}{1000\cdot(1+s)^2}
    \]
    Its first derivative is given by
    \[
        F'(s) = \frac{-298 \cdot s^6 - 447 \cdot s^5 - 550 \cdot s^4 - 1100 \cdot s^3 + 1543 \cdot s - 156}{500 \cdot (1+s)^3},
    \]
    which has two roots on the interval $[0,1]$, namely $s_0 \approx 0.1019$ and $s_1 \approx 0.8194$. Moreover, $F'(s) < 0$ for $s \in [0, s_0)$,  $F'(s) > 0$ for $s \in (s_0, s_1)$, and $F'(s) < 0$ for $s \in (s_1,1]$. 
    Since $F(s_1) > F(0)$, we obtain that $F$ achieves its maximum in the interval $[0,1]$ at $s_1$. 

    Now, for a fixed $q \in (0,1]$, we define a function 
    \[ 
        G_q(x) \define \frac{h(q) \cdot xq}{(1+xq) \cdot (1+x)}.
    \]
    The first derivative of $G_q(x)$ is 
    \[ 
        G'_q(x) = h(q) \cdot q \cdot \frac{(1+xq)(1+x)-x(1+q+2qx)}{(1+xq)^2 \cdot (1+x)^2}=\frac{1-qx^2}{(1+xq)^2 \cdot (1+x)^2},
    \]
    and its only positive root is given by $x=1/\sqrt{q}$. As $G_q(x)$ is defined for all $x \in [0,\infty)$, $G_q(0)=0$, $\lim_{x\rightarrow \infty} G(x)=0$, and $G_q(1/\sqrt{q}) >0$, we obtain that $G_q(x)$ attains its maximum on $[0,\infty)$ at $x=1/\sqrt{q}$.
    Therefore, 
    \begin{align*}
        \max_{q \in (0,1],\, x \geq 0} G_q(x) 
        & = \max_{q \in (0,1]} G_q\left(1/\sqrt{q}\right) 
        = \max_{q \in (0,1]} \frac{h(q) \cdot \sqrt{q}}{(1+\sqrt{q}) \cdot (1+1/\sqrt{q})} \\
        & = \max_{q \in (0,1]} \frac{h(q) \cdot q}{(1+\sqrt{q})^2} 
        = \max_{s \in (0,1]} \frac{h(s^2) \cdot s^2}{(1+s)^2} \\
        & = \max_{s \in (0,1]} F(s) = F(s_1)
    \end{align*}
    Finally, we conclude that
    \begin{align*}
        \max_{q \in (0,1], x \geq 0}  \left(1 - G_q(x)\right)^{-1} 
        & = \left(1 - \max_{q \in (0,1]} G_q\left(\frac{1}{\sqrt{q}}\right)\right)^{-1} \\
        & = \left(1 - F(s_1) \right)^{-1} \\
        & \approx 1.37282 < 1.3729 = \rho.
        \qedhere
    \end{align*}
\end{proof}